\documentclass[12pt,reqno]{amsart}%
\usepackage{graphicx}
\usepackage{amscd}
\usepackage{amsmath}
\usepackage{amsfonts}
\usepackage{amssymb}%
\providecommand{\U}[1]{\protect\rule{.1in}{.1in}}
\numberwithin{equation}{section}
\providecommand{\U}[1]{\protect\rule{.1in}{.1in}}
\providecommand{\U}[1]{\protect\rule{.1in}{.1in}}
\allowdisplaybreaks

\theoremstyle{plain}

\newtheorem{lemma}{Lemma}

\def\Ai{\operatorname{Ai}}
\begin{document}
\title[Fundamental Laser Modes]{Fundamental Laser Modes in Paraxial Optics: \\
From Computer Algebra and Simulations \\
to Experimental Observation}
\author{Christoph~Koutschan}
\address{Johann Radon Institute for Computational and Applied Mathematics, Austrian
Academy of Sciences, Altenberger Stra\ss e 69, A-4040 Linz, Austria}
\email{christoph.koutschan@ricam.oeaw.ac.at}
\author{Erwin~Suazo}
\address{School of Mathematical and Statistical Sciences, University of Texas of Rio Grande Valley,
1201~W.~University Drive, Edinburg, TX 78541, U.S.A.}
\email{erwin.suazo@utrgv.edu}
\author{Sergei~K.~Suslov}
\address{School of Mathematical and Statistical Sciences, Arizona State University, Tempe, AZ 85287--1804, U.S.A.}
\email{sks@asu.edu}
\urladdr{http://hahn.la.asu.edu/\symbol{126}suslov/index.html}
\date{\today}
\keywords{Paraxial wave equation, Green's function, generalized Fresnel integrals,
Airy-Hermite-Gaussian beams, Hermite-Gaussian beams, Laguerre-Gaussian beams,
Bessel-Gaussian beams, spiral beams, time-dependent Schr\"{o}dinger equation.\\
\phantom{Ga}The final publication is available at link.springer.com (Applied Physics B, DOI 10.1007/s00340-015-6231-9).}

\begin{abstract}
We study multi-parameter solutions of the inhomogeneous paraxial wave equation
in a linear and quadratic approximation which include oscillating laser beams
in a parabolic waveguide, spiral light beams, and other important families of
propagation-invariant laser modes in weakly varying media. A \textquotedblleft
smart\textquotedblright\ lens design and a similar effect of superfocusing of
particle beams in a thin monocrystal film are also discussed. In the
supplementary electronic material, we provide a computer algebra verification
of the results presented here, and of some related mathematical tools that
were stated without proofs in the literature.  We also demonstrate how
computer algebra can be used to derive some of the presented formulas
automatically, which is highly desirable as the corresponding hand
calculations are very tedious.  In numerical simulations, some of the new
solutions reveal quite exotic properties which deserve further investigation
including an experimental observation.
\end{abstract}
\maketitle

\section{Introduction}
\label{sec:intro}
In this article, we study multi-parameter laser modes in (linear) paraxial
optics with the help of computer algebra methods by using an analogy with
quantum mechanics. In particular, the Ermakov-type system's approach to
generalized quantum harmonic oscillators is utilized to paraxial, or
parabolic, wave equations in a weakly inhomogeneous lens-like medium. Although
several different techniques are widely available for integrating the
(scalar) parabolic equations (see, for instance, contemporary reviews
\cite{AbramAlievaRodrigo12}, \cite{KogelnikLi66}, \cite{SiegmanLasers}, \cite{Steinberg77},
\cite{Torre13}, \cite{YarivYeh} and the references therein), we would like to
explore a variant of the Fresnel integral and a certain generalization of the
lens transformation \cite{Mah:Sua:Sus13} combined with explicit solutions of
the Ermakov-type system introduced in \cite{Lan:Lop:Sus} as an alternative to
the traditional $ABCD$ law of paraxial optics \cite{Kogelnik65},
\cite{KogelnikLi66}.

Knowledge of solutions of this system is particularly
useful as they describe the propagation of Gaussian beams through various
optical elements and devices that are subject to a noisy environment. We
demonstrate that this approach gives a natural mathematical description of
other special laser modes in optical systems under consideration. In the
spirit of a modern \textquotedblleft doing science by a
computer\textquotedblright\ paradigm, a computer algebra derivation of all
main results is given in the form of a \textsl{Mathematica}
notebook~\cite{notebook}, with the aid of algorithmic tools presented in
\cite{Koutschan09}, \cite{Koutschan10b}, \cite{Koutschan13}.
One specific aim of this paper is to advertise and propagate the usage of
computer algebra methods in physics. Our intention is to convince the reader
that, in many instances, tedious hand calculations and longish proofs can be
replaced by \textquotedblleft pressing a button\textquotedblright.

For example, today's computer algebra systems---like \textsl{Mathematica}
which we utilize in this work---are powerful enough to verify that the
expression
\begin{align*}
  \psi(x,t) &= \sqrt{\frac{\beta }{1+4\alpha t}}\ \exp\left(  i\,\frac
    {\alpha x^{2}+\delta x-\delta^{2} t}{1+4\alpha t} +
    \frac{i\beta^{2} t}{1+4\alpha t}
    \left( \varepsilon +\frac{\beta x-2\beta \delta t}{1+4\alpha t}
    -\frac{2}{3}\frac{\beta^{4} t^{2}}{(1+4\alpha t)^{2}}\right) \right) \\
  &  \quad\times\Ai\left(  \varepsilon +\frac{\beta
    x-2\beta \delta t}{1+4\alpha t}-\frac{\beta^{4} t^{2}}{(1+4\alpha t)^{2}}\right),
\end{align*}
which corresponds to Equation~\eqref{MultiAiry} below, is a particular
solution of the time-dependent Schr\"{o}dinger equation $i\psi_{t}+\psi_{xx}=0$.
Note that this verification is done on a completely symbolic and rigorous
level, working with the symbolic parameters $\alpha$, $\beta$, $\delta$,
and~$\varepsilon$. In this instance, it suffices to enter the definition of
$\psi(x,t)$ and the expression $i\psi_t+\psi_{xx}$ into the computer and employ the
command \texttt{FullSimplify}; the result is~$0$, as expected. In other
examples one needs to assist the computer algebra system in the simplification
task, as explained in the main part of the paper.  A different approach, which
is also presented here, consists in automatically deriving a differential
equation satisfied by a given expression. In our example, this means that from
the input expression $\psi(x,t)$ above, the computer is able to derive the
equation $i\psi_t+\psi_{xx}=0$, using algorithms implemented in the
\texttt{HolonomicFunctions} package~\cite{Koutschan10b}; more details on this
example are given in Section~\ref{sec:CA}.

In general, (multi-parameter) laser modes describe the beam deviation
from the optical axis and an oblique propagation in the optical system,
which is usually not discussed in the literature.
They also correspond to the maximal kinematical invariance groups of the paraxial wave equations.
We believe that all these new features may help the reader to use this approach in numerical simulations,
as well as in the future experiment design and observations.
The computer codes for all fundamental laser modes in question are provided,
which may stimulate a new interest to this classical topic in paraxial optics.

For a more traditional approach to the paraxial wave equations and for their numerous
applications in optics and engineering, the reader can be referred to the
classical accounts \cite{AbramAlievaRodrigo12}, \cite{AbramVolBook},
\cite{AhmanNikBook}, \cite{BandresGuiz09}, \cite{BornWolf}, \cite{Brosseau09},
\cite{Eichmann71}, \cite{Fock65}, \cite{GoncharenkoBook}, \cite{KogelnikLi66},
\cite{SiegmanLasers}, \cite{Smith00}, \cite{TuruFrieb09}, \cite{Yariv3rdEd},
\cite{YarivYeh}, \cite{VainshBook}, \cite{VinRudSuxBook79},
\cite{VlasovTalanovParabolicEq}, \cite{WaltherLENSES}. The interested reader is referred
to\ \cite{AbramVolBook}, \cite{DeutschGarrison91}, \cite{Fock65},
\cite{GurevichNLIonosphere}, \cite{Laxetal75}, \cite{MarteStenholm97} for
further details on the transition from Maxwell to paraxial wave optics; see
also \cite{AbramVol89}, \cite{AbramVol89a}, \cite{KravOrlBook},
\cite{Timofeev05}, and \cite{WrightGarr87} for different aspects of
geometrical optics and paraxial approximation in the diffraction theory. Among
contemporary topics, a modern status of the concept of photon, second
quantization, photon spin, helicity, and angular momentum are discussed in \cite{BB06},
\cite{BB11}, \cite{HatzviSchechner12Suppl}, \cite{HausPan93}, \cite{Jack2nd}, \cite{Kr:Sus12},
\cite{Kr:Lan:Sus15}, \cite{Okulov08}; see also the references therein.
Gaussian light packages, that are highly localized in space and time, are investigated
in \cite{Kiselev07}, \cite{KiselevPerel}, \cite{Kiselevetal12}.

The article is organized as follows: In the next section, we discuss basics of
our approach, including the Green's Function and Fresnel Integrals, as well as
the derivation of the above-mentioned Ermakov-type system. This is followed,
in Section~\ref{sec:multi}, by a review of various multi-parameter laser modes
and some of their applications. In Section~\ref{sec:CA} we explain the
computer algebra tools employed in this article. Useful vectorial solutions of
Maxwell's equations in paraxial approximation are derived in appendices.  Our
numerical simulations, see~\cite{notebook}, reveal some quite exotic
properties of the laser modes under consideration. We believe that they
deserve an experimental verification in some spectacular $2D$ and $3D$
experiments on the beam propagation.

\section{Green's Function and Fresnel Integrals for Inhomogeneous Media}
\label{sec:basics}

This section comprises a brief survey of results established in
\cite{Cor-Sot:Lop:Sua:Sus}, \cite{Dodonov:Man'koFIAN87}, \cite{Lan:Lop:Sus},
\cite{Mah:Sua:Sus13}, \cite{Mah:Sus13}, \cite{SuazoSuslovSol}, \cite{Suslov11}
(see also the references therein for the classical accounts) which are
composed here in a compact form in order to make our presentation as
self-contained as possible. In addition, we present independent proofs in the
supplementary electronic material~\cite{notebook} for the reader's benefits.
In the context of paraxial optics, this approach, among other things, allows
one to unify various laser modes introduced and studied by different authors
(a detailed bibliography is provided below but we apologize in advance if an
important reference is missing).

\subsection{Unidimensional Case}

Recent advances in quantum mechanics of generalized harmonic oscillators can
be utilized in order to solve similar problems concerning the light
propagation in a general lens-like medium \cite{Dodonov:Man'koFIAN87},
\cite{Eichmann71}, \cite{Kogelnik65}, \cite{KogelnikLi66},
\cite{KrivoshSissak80}, \cite{KrivoshSauter92}, \cite{Mah:Sua:Sus13},
\cite{Mah:Sus13}, \cite{Sieg73}.

\subsubsection{Green's Function and Generalized Fresnel Integrals}

In the context of quantum mechanics, the $1D$ linear Schr\"{o}dinger equation
for generalized driven harmonic oscillators,%
\begin{align}
&  i\psi_{t}=-a(t) \psi_{xx}+b(t) x^{2}\psi-ic(t) x\psi_{x}\nonumber\\
&  \qquad-id(t) \psi-f(t) x\psi+ig(t) \psi_{x} \label{SchroedingerQuadratic}%
\end{align}
($a,$ $b,$ $c,$ $d,$ $f,$ and $g$ are suitable real-valued functions of the
time $t$\ only), can be solved by the integral superposition principle:%
\begin{equation}
\psi(x,t) =\int_{-\infty}^{\infty}G(x,y,t) \psi(y,0) \ dy,
\label{Superposition}%
\end{equation}
where Green's function $G(x,y,t)$ is given by
\begin{equation}
\label{Green1D}G(x,y,t) = \frac{1}{\sqrt{2\pi i\mu_{0}(t)}} \exp
\Big(i\big( \alpha_{0}(t) x^{2}+\beta_{0}(t) xy+\gamma_{0}(t) y^{2}+\delta
_{0}(t) x+\varepsilon_{0}(t) y+\kappa_{0}(t) \big)\Big)
\end{equation}
for suitable initial data $\psi(x,0) =\varphi(x) $ (see
\cite{Cor-Sot:Lop:Sua:Sus}, \cite{Lan:Lop:Sus}, \cite{Suslov11} and the
references therein for more details).

The functions $\alpha_{0}, \beta_{0}, \gamma_{0}, \delta_{0}, \varepsilon_{0}%
$, and $\kappa_{0}$ are given by \cite{Cor-Sot:Lop:Sua:Sus}, \cite{Suslov11}:%
\begin{align}
\alpha_{0}(t)  &  = \frac{1}{4a(t)}\frac{\mu_{0}^{\prime}(t)}{\mu_{0}(t)}
-\frac{d(t)}{2a(t) },\label{A0}\\
\beta_{0}(t)  &  = -\frac{\lambda(t)}{\mu_{0}(t)}, \qquad\lambda
(t)=\exp\left(  -\int_{0}^{t}\big(c(s)-2d(s)\big) \ ds\right)  ,\label{B0}\\
\gamma_{0}(t)  &  = \frac{1}{2\mu_{1}(0) }\frac{\mu_{1}(t) }{\mu_{0}(t)
}+\frac{d(0) }{2a(0) } \label{C0}%
\end{align}
and%
\begin{align}
\delta_{0}(t)  &  = \frac{\lambda(t) }{\mu_{0}(t) }\int_{0}^{t}\left(  \left(
f(s) -\frac{d(s) }{a(s) }g(s) \right)  \mu_{0}(s) +\frac{g(s) }{2a(s) }\mu
_{0}^{\prime}(s) \right)  \frac{ds}{\lambda(s) },\label{D0}\\
\varepsilon_{0}(t)  &  = -\frac{g(0) }{2a(0) }+2\int_{0}^{t}\frac{\lambda(s)
\big( a(s) \left(  f(s) -\delta_{0}^{\prime}(s) \right)  -d(s) g(s)
\big) +a(s) \delta_{0}(s) \lambda^{\prime}(s) }{\mu_{0}^{\prime}(s)
}\ ds\label{E0}\\
&  =-\frac{2a(t) \lambda(t) }{\mu_{0}^{\prime}(t) }\delta_{0}(t) +8\int%
_{0}^{t}\frac{a(s) \sigma(s) \lambda(s) }{\left(  \mu_{0}^{\prime}(s) \right)
^{2}}\mu_{0}(s)\delta_{0}(s) \ ds\nonumber\\
&  \qquad+2\int_{0}^{t}\frac{a(s) \lambda(s) } {\mu_{0}^{\prime}(s) }\left(
f(s) -\frac{d(s) }{a(s) }g(s) \right)  \ ds,\nonumber\\
\kappa_{0}(t)  &  = \int_{0}^{t}\delta_{0}(s) \big( g(s) -a(s) \delta_{0}(s)
\big) \ ds\label{F0}\\
&  = \frac{a(t) \mu_{0}(t) }{\mu_{0}^{\prime}(t) }\delta_{0}^{2}(t) -4\int%
_{0}^{t}\frac{a(s) \sigma(s) }{\left(  \mu_{0}^{\prime}(s) \right)  ^{2}%
}\left(  \mu_{0}(s) \delta_{0}(s) \right)  ^{2}\ ds\nonumber\\
&  \qquad-2\int_{0}^{t}\frac{a(s) }{\mu_{0}^{\prime}(s) } \mu_{0}(s)
\delta_{0}(s) \left(  f(s) -\frac{d(s) }{a(s) }g(s) \right)  \ ds\nonumber
\end{align}
provided that $\mu_{0}$ and $\mu_{1}$ are the standard (real-valued) solutions
of the characteristic equation:%
\begin{equation}
\mu^{\prime\prime}(t)-\tau(t) \mu^{\prime}(t)+4\sigma(t)\mu(t)=0
\label{CharEq}%
\end{equation}
with varying coefficients%
\begin{equation}
\tau(t) =\frac{a^{\prime}}{a}-2c+4d,\qquad\sigma(t) =ab-cd+d^{2}+\frac{d}%
{2}\left(  \frac{a^{\prime}}{a}-\frac{d^{\prime}}{d}\right)  ,
\label{TauSigma}%
\end{equation}
subject to the initial conditions $\mu_{0}(0) =0,$ $\mu_{0}^{\prime}(0) =2a(0)
\neq0$ and $\mu_{1}(0) \neq0,$ $\mu_{1}^{\prime}(0) =0.$ The Wronskian of
these standard solutions is given by%
\begin{equation}
W(\mu_{0},\mu_{1}) = \mu_{0}\mu_{1}^{\prime}- \mu_{0}^{\prime}\mu_{1} =
-2\mu_{1}(0) a(t) \lambda^{2}(t). \label{Wronskian}%
\end{equation}
Our coefficients (\ref{A0})--(\ref{F0}) satisfy the so-called Riccati-type
system, see the unidimensional case of Equations \eqref{SysA}--\eqref{SysF}
below with $c_{0}=0$ \cite{Lan:Lop:Sus}, subject to the following asymptotic
expansions%
\begin{align}
\alpha_{0}(t)  &  = \frac{1}{4a(0) t}-\frac{c(0) }{4a(0) }-\frac{a^{\prime}(0)
}{8a^{2}(0) }+\mathcal{O}(t) ,\label{CoeffAsymps}\\
\beta_{0}(t)  &  = -\frac{1}{2a(0) t}+\frac{a^{\prime}(0) }{4a^{2}(0)
}+\mathcal{O}(t) ,\nonumber\\
\gamma_{0}(t)  &  = \frac{1}{4a(0) t}+\frac{c(0) }{4a(0) }-\frac{a^{\prime}(0)
}{8a^{2}(0) }+\mathcal{O}(t) ,\nonumber\\
\delta_{0}(t)  &  = \frac{g(0) }{2a(0) }+\mathcal{O}(t) ,\nonumber\\
\varepsilon_{0}(t)  &  = -\frac{g(0) }{2a(0) }+ \mathcal{O}(t) ,\nonumber\\
\kappa_{0}(t)  &  = \mathcal{O}(t)\nonumber
\end{align}
as $t\rightarrow0.$ As a result,%
\begin{align}
G(x,y,t)  &  \sim\frac{1}{\sqrt{2\pi ia(0) t}}\exp\left(  i\frac{\left(
x-y\right)  ^{2}}{4a(0) t}\right) \label{GreenAsymp}\\
&  \times\exp\left(  -i\left(  \frac{a^{\prime}(0) }{8a^{2}(0) }\left(
x-y\right)  ^{2}+\frac{c(0) }{4a(0) }\left(  x^{2}-y^{2}\right)  -\frac{g(0)
}{2a(0) }\left(  x-y\right)  \right)  \right)  .\nonumber
\end{align}
Here, $f\sim g$ as $t\rightarrow0,$ if $\lim_{t\rightarrow0}\left(
f/g\right)  =$ $1.$ (For applications, say to random media
(\cite{RytovetalPrinciples}, \cite{TangMah13}), the integrals are treated in
the most general way which includes stochastic calculus; see, for example,
\cite{Oksendal00}.)

\textit{Note.} Most of these results were only stated in the original
publications because its detailed calculations are pretty messy and
time-consuming without use of algorithmic tools. In this article, for the
reader's benefits we present systematic computer algebra proofs of these
results~\cite{notebook}.

In the context of paraxial optics, when the time variable $t$ represents the
coordinate in the direction of the system optical axis, say $s,$ for wave propagation, the
expressions (\ref{Superposition})--(\ref{Green1D}) can be thought of as a
generalization of Fresnel integrals \cite{AbramVolUFN}, \cite{AbramVolBook},
\cite{BornWolf}, \cite{Dodonov:Man'koFIAN87}, \cite{GoncharenkoBook},
\cite{KrivoshSissak80}, \cite{KrivoshPetrovSisak85},
\cite{KrivoshPetrovSisak86}, \cite{KrivoshSauter92}, \cite{Mah:Sus13},
\cite{VinRudSuxBook79}. The corresponding Schr\"{o}dinger
equation~\eqref{SchroedingerQuadratic}, with $t\rightarrow s,$ can be referred
to as a generalized paraxial or parabolic wave equation \cite{Mah:Sua:Sus13},
\cite{Mah:Sus13}.

\subsubsection{Special Beam Modes in Weakly Inhomogeneous Media}

An important particular solution (generalized Hermite-Gaussian beams in
optics) of the parabolic equation \eqref{SchroedingerQuadratic} is given by
\cite{Lan:Lop:Sus}:%
\begin{equation}
\psi_{n}(x,s) =\frac{e^{i\left(  \alpha x^{2}+\delta x+\kappa\right)
+i\left(  2n+1\right)  \gamma}}{\sqrt{2^{n}n!\mu\sqrt{\pi}}}\ e^{-\left(
\beta x+\varepsilon\right)  ^{2}/2}\ H_{n}(\beta x+\varepsilon) ,
\label{WaveFunction}%
\end{equation}
where $H_{n}(x) $ are the Hermite polynomials \cite{Ni:Su:Uv}. Here,%
\begin{align}
\mu &  =\mu(0) \mu_{0}\sqrt{\beta^{4}(0) +4\left(  \alpha(0) +\gamma
_{0}\right)  ^{2}},\label{MKernelOsc}\\
\alpha &  =\alpha_{0}-\beta_{0}^{2}\frac{\alpha(0) +\gamma_{0}}{\beta^{4}(0)
+4\left(  \alpha(0) +\gamma_{0}\right)  ^{2}},\label{AKernelOsc}\\
\beta &  =-\frac{\beta(0) \beta_{0}}{\sqrt{\beta^{4}(0) +4\left(  \alpha(0)
+\gamma_{0}\right)  ^{2}}}=\frac{\beta(0) \mu(0) }{\mu(t) }\lambda(t)
,\label{BKernelOsc}\\
\gamma &  =\gamma(0) -\frac{1}{2}\arctan\frac{\beta^{2}(0) }{2\left(
\alpha(0) +\gamma_{0}\right)  },\quad a(0) >0\label{CKernelOsc}\\
\delta &  =\delta_{0}-\beta_{0}\frac{\varepsilon(0) \beta^{3}(0) +2\left(
\alpha(0) +\gamma_{0}\right)  \left(  \delta(0) +\varepsilon_{0}\right)
}{\beta^{4}(0) +4\left(  \alpha(0) +\gamma_{0}\right)  ^{2}}%
,\label{DKernelOsc}\\
\varepsilon &  =\frac{2\varepsilon(0) \left(  \alpha(0) +\gamma_{0}\right)
-\beta(0) \left(  \delta(0) +\varepsilon_{0}\right)  }{\sqrt{\beta^{4}(0)
+4\left(  \alpha(0) +\gamma_{0}\right)  ^{2}}},\label{EKernelOsc}\\
\kappa &  =\kappa(0) +\kappa_{0}-\varepsilon(0) \beta^{3}(0) \frac{\delta(0)
+\varepsilon_{0}}{\beta^{4}(0) +4\left(  \alpha(0) +\gamma_{0}\right)  ^{2}%
}\label{FKernelOsc}\\
&  \qquad+\left(  \alpha(0) +\gamma_{0}\right)  \frac{\varepsilon^{2}(0)
\beta^{2}(0) -\left(  \delta(0) +\varepsilon_{0}\right)  ^{2}}{\beta^{4}(0)
+4\left(  \alpha(0) +\gamma_{0}\right)  ^{2}}\nonumber
\end{align}
in terms of the fundamental solution subject to the arbitrary real or
complex-valued initial data $\mu(0) \neq0,$ $\alpha(0) ,$ $\beta(0) \neq0,$
$\gamma(0) ,$ $\delta(0) ,$ $\varepsilon(0) ,$ $\kappa(0)$. This solution was
obtained in \cite{Lan:Lop:Sus} by an integral evaluation and its direct
verification by substitution is provided in~\cite{notebook}.

\textit{Note.} Equations (\ref{AKernelOsc})--(\ref{FKernelOsc}) solve the
one-dimensional case of the Ermakov-type system (\ref{SysA})--(\ref{SysF})
below with $c_{0}=1$ \cite{Lan:Lop:Sus}; for the complex form of these
solutions, see \cite{Kr:Sus12}; their verification is provided
in~\cite{notebook}.

By the superposition principle, (orthonormal) solutions (\ref{WaveFunction})
can be used for the corresponding eigenfunction expansions in the case of
real-valued initial data. In our approach, the functions $f$ and $g$ are
treated as two stochastic processes and Equations~\eqref{D0}--\eqref{F0} and
(\ref{DKernelOsc})--(\ref{FKernelOsc}) can be analyzed by statistical methods
\cite{AhDyaChirBook81}, \cite{RytovetalPrinciples} (which may include random
initial data).

A solution in terms of Airy functions \cite{Fock65} (generalized Airy beams)
has the form \cite{Mah:Sus12}, \cite{Mah:Sus13}:%
\begin{equation}
\psi(x,s)=\frac{e^{i\left(  \alpha x^{2}+\delta x+\kappa\right)  -i\left(
\beta x+\varepsilon-2\gamma^{2}/3\right)  \gamma}}{\sqrt{\mu}}%
\ \Ai(\beta x+\varepsilon-\gamma^{2}), \label{AiryOsc1D}%
\end{equation}
where
\begin{align}
\mu &  =2\mu(0)\mu_{0}\left(  \alpha(0)+\gamma_{0}\right)  ,\label{MKernel}\\
\alpha &  =\alpha_{0}-\frac{\beta_{0}^{2}}{4\left(  \alpha(0)+\gamma
_{0}\right)  },\label{AKernel}\\
\beta &  =-\frac{\beta(0)\beta_{0}}{2\left(  \alpha(0)+\gamma_{0}\right)
}=\frac{\beta(0)\mu(0)}{\mu}\lambda,\label{BKernel}\\
\gamma &  =\gamma(0)-\frac{\beta^{2}(0)}{4\left(  \alpha(0)+\gamma_{0}\right)
},\label{CKernel}\\
\delta &  =\delta_{0}-\frac{\beta_{0}\left(  \delta(0)+\varepsilon_{0}\right)
}{2\left(  \alpha(0)+\gamma_{0}\right)  },\label{DKernel}\\
\varepsilon &  =\varepsilon(0)-\frac{\beta(0)\left(  \delta(0)+\varepsilon
_{0}\right)  }{2\left(  \alpha(0)+\gamma_{0}\right)  },\label{EKernel}\\
\kappa &  =\kappa(0)+\kappa_{0}-\frac{\left(  \delta(0)+\varepsilon
_{0}\right)  ^{2}}{4\left(  \alpha(0)+\gamma_{0}\right)  }. \label{FKernel}%
\end{align}
A direct verification is given in~\cite{notebook} for the reader's benefits.
Important special cases of Airy beams were found in \cite{BerryBalazs79},
\cite{SiviloglouChris07}, and \cite{Siviloglouetal07} (see also
\cite{Mah:Sus12}, \cite{Torre13} and the references therein; more details are
given in Section~\ref{sec:AiryBeams} below).

\textit{Note.} Equations (\ref{MKernel})--(\ref{FKernel}) solve the
one-dimensional case of the Riccati-type system (\ref{SysA})--(\ref{SysF})
below with $c_{0}=0$ \cite{Lan:Lop:Sus}; a proof is provided
in~\cite{notebook}. Moreover, in view of uniqueness of the Cauchy initial
value problem for Schr\"{o}dinger equation (\ref{SchroedingerQuadratic}), the
use of Green's function (\ref{Green1D}) in Equation (\ref{Superposition})
results in an integral evaluation for Airy functions which may have an
independent value.

In general, one may interpret solutions (\ref{WaveFunction}%
)--(\ref{FKernelOsc}) and (\ref{AiryOsc1D})--(\ref{FKernel}), relating the
initial and final parameters of the corresponding laser modes propagating in a
certain element of optical device, as an alternative to the $ABCD$ law which
follows from the analogy between the laws for laser beams and the laws obeyed
by the spherical waves in geometrical optics \cite{Kogelnik65},
\cite{KogelnikLi66}. As one can see, the corresponding composition laws will include
a variant of linear fractional transformation when $\alpha(0)\neq0.$
A numerical example is discussed in Section~\ref{sec:smart}.
(Further details of this interpretation are left to the reader.)

\subsection{Two-Dimensional Case}

For the laser beam propagation in optics, the (co-dimensional) $2D$ case (with or without
cylindrical symmetry) is of a great importance.

\subsubsection{Separation of Variables}

In the paraxial approximation, a $2D$ coherent light field in a general
lens-like medium with coordinates $(\boldsymbol{r},s)=(x,y,s)$ can be
described by the following equation for the complex field amplitude:%
\begin{equation}
i\psi_{s}(\boldsymbol{r},s)=H\psi(\boldsymbol{r},s),\qquad H=H_{1}%
(x,s)+H_{2}(y,s), \label{2DLinearSchroedingerEquation}%
\end{equation}
where $H_{1,2}$ are the Hamiltonians in $x$ and $y$ directions similar to one
in (\ref{SchroedingerQuadratic}) but, in a general inhomogeneous medium model,
with two different sets of suitable functions $a_{1,2}(s),$ $b_{1,2}(s),$
$c_{1,2}(s),$ $d_{1,2}(s),$ $f_{1,2}(s),$ and $g_{1,2}(s).$ (We assume, for
simplicity, that the nondiagonal terms are eliminated by passing to normal
coordinates.) The kernel of generalized Fresnel integral can be obtained as
the product \cite{Mah:Sus13}:%
\begin{equation}
G(\boldsymbol{r},\boldsymbol{r}^{\prime},s)=G_{1}(x,\xi,s)G_{2}(y,\eta,s),
\label{2DGreenFunction}%
\end{equation}
where the kernels $G_{1,2}$ are given by (\ref{Green1D}) with a simple change
of notation: the coefficients $\alpha_{0}^{(1,2)},$ $\beta_{0}^{(1,2)},$
$\gamma_{0}^{(1,2)},$ $\delta_{0}^{(1,2)},$ $\varepsilon_{0}^{(1,2)},$
$\kappa_{0}^{(1,2)}$ are defined, in general, in terms of two sets of the
fundamental solutions (\ref{A0})--(\ref{F0}) with $t\leftrightarrow s.$ The
solution of the corresponding boundary value problem can be found by the
integral superposition principle ($2D$ generalized Fresnel integral):%
\begin{equation}
\psi(\boldsymbol{r},s)=\iint_{\mathbb{R}^{2}}G(\boldsymbol{r},\boldsymbol{r}%
^{\prime},s)\psi(\boldsymbol{r}^{\prime},0)\ d\boldsymbol{r}^{\prime}
\label{2DFresnelIntegral}%
\end{equation}
for suitable initial data. (This integral determines the spatial beam
evolution during the Fresnel diffraction.)

The corresponding $2D$ Hermite-Gaussian beams have the form
\begin{align}
&  \psi_{nm}(\boldsymbol{r},s)=\frac{e^{i\left(  \kappa_{1}+\kappa_{2}\right)
}}{\sqrt{2^{n+m}n!m!\mu^{(1)}\mu^{(2)}\pi}}e^{i\left(  \alpha_{1}x^{2}%
+\delta_{1}x\right)  +i\left(  2n+1\right)  \gamma_{1}}\ e^{i\left(
\alpha_{2}y^{2}+\delta_{2}y\right)  +i\left(  2m+1\right)  \gamma_{2}%
}\label{2DGaussHermite}\\
&  \quad\quad\qquad\times e^{-\left(  \beta_{1}x+\varepsilon_{1}\right)
^{2}/2-\left(  \beta_{2}y+\varepsilon_{2}\right)  ^{2}/2}\ \ H_{n}\left(
\beta_{1}x+\varepsilon_{1}\right)  H_{m}\left(  \beta_{2}y+\varepsilon
_{2}\right) \nonumber
\end{align}
in terms of solutions of the Ermakov-type system (\ref{SysA})--(\ref{SysF})
below with $c_{0}=1,$ which are known in quadratures \cite{Lan:Lop:Sus} (see
also (\ref{hhA})--(\ref{hhK}) for an important explicit special case).
Equations \eqref{MKernelOsc}--\eqref{FKernelOsc} are valid with a similar
change of notation for given initial data $\mu^{(1,2)}(0),$ $\alpha_{1,2}(0),$
$\beta_{1,2}(0)\neq0,$ $\gamma_{1,2}(0),$ $\delta_{1,2}(0),$ $\varepsilon
_{1,2}(0),$ $\kappa_{1,2}(0)$ (see also \cite{Agrawaletal74},
\cite{AhDyaChirBook81}, \cite{GoncharenkoBook}, \cite{Sieg73},
\cite{Yariv3rdEd}, \cite{YarivYeh}, \cite{VinRudSuxBook79} for various special cases).

In general, by the separation of variables, the product of any two $1D$
solutions (\ref{WaveFunction}) and (\ref{AiryOsc1D}), say%
\begin{equation}
\psi_{n}(\boldsymbol{r},s)=\psi_{n}(x,s)\psi(y,s), \label{AiryHermite}%
\end{equation}
gives an important class of $2D$ solutions (Airy-Hermite-Gaussian beams in a
weakly inhomogeneous medium; see also \cite{GuGbur10}, \cite{Gu13},
\cite{GburVisser03}).

\subsubsection{Cylindrical Symmetry}

If $a_{1}(s)=a_{2}(s)=a(s),$ $b_{1}(s)=b_{2}(s)=b(s),$ $c_{1}(s)=c_{2}%
(s)=c(s),$ $d_{1}(s)=d_{2}(s)=d(s),$ the parabolic equation,%
\begin{align}
&  iA_{s}=-a\left(  A_{xx}+A_{yy}\right)  +b\left(  x^{2}+y^{2}\right)
A-ic\left(  xA_{x}+yA_{y}\right) \label{2DSchroedingerNL}\\
&  \qquad-2idA-\left(  xf_{1}+yf_{2}\right)  A+i\left(  g_{1}A_{x}+g_{2}%
A_{y}\right)  ,\nonumber
\end{align}
where $f_{1,2}(s)$ and $g_{1,2}(s)$ are real-valued functions of a coordinate
in the direction of the optical axis~$s$ related to the wave propagation, can
be reduced to the standard forms%
\begin{equation}
-i\chi_{\tau}+\chi_{\xi\xi}+\chi_{\eta\eta}=c_{0}\left(  \xi^{2}+\eta
^{2}\right)  \chi,\qquad\left(  c_{0}=0,1\right)
\label{2DSchroedingerTransformed}%
\end{equation}
by the following ansatz%
\begin{equation}
A=\mu^{-1}e^{i\left(  \alpha\left(  x^{2}+y^{2}\right)  +\delta_{1}%
x+\delta_{2}y+\kappa_{1}+\kappa_{2}\right)  }\ \chi(\xi,\eta,\tau)
\label{Ansatz2DParabolic}%
\end{equation}
(see Lemma~1 of \cite{Mah:Sus13}, which is reproduced below in our notation
with an independent computer algebra proof for the reader's convenience).

\begin{lemma}
\label{lem:nonlinear} The nonlinear parabolic equation,%
\begin{align}
iA_{s}  &  =-a\left(  A_{xx}+\psi_{yy}\right)  +b\left(  x^{2}+y^{2}\right)
A-ic\left(  xA_{x}+yA_{y}\right)  -2idA\label{NonlinearParabolic2D}\\
&  -\left(  xf_{1}+yf_{2}\right)  A+i\left(  g_{1}A_{x}+g_{2}A_{y}\right)
+h\left\vert A\right\vert ^{p}A,\nonumber
\end{align}
where $a,$ $b,$ $c,$ $d,$ $f_{1,2}$ and $g_{1,2}$ are real-valued functions of
$s,$ can be transformed to%
\begin{equation}
-i\chi_{\tau}+\chi_{\xi\xi}+\chi_{\eta\eta}=c_{0}\left(  \xi^{2}+\eta
^{2}\right)  \chi+h_{0}\left\vert \chi\right\vert ^{p}\chi\qquad\left(
c_{0}=0,1\right)  \label{2DShroedingerNLTransformed}%
\end{equation}
by the ansatz (\ref{Ansatz2DParabolic}), where $\xi=\beta(s)x+\varepsilon
_{1}(s),$ $\eta=\beta(s)y+\varepsilon_{2}(s),$ $\tau=\gamma(s),$
$h=h_{0}a\beta^{2}\mu^{p}$\ $(h_{0}$ is a constant), provided that%
\begin{align}
\frac{d\alpha}{ds}+b+2c\alpha+4a\alpha^{2}  &  =c_{0}a\beta^{4},\label{SysA}\\
\frac{d\beta}{ds}+(c+4a\alpha)\beta &  =0,\label{SysB}\\
\frac{d\gamma}{ds}+a\beta^{2}  &  =0,\label{SysC}\\
\frac{d\delta_{1,2}}{ds}+(c+4a\alpha)\delta_{1,2}  &  =f_{1,2}+2g\alpha
+2c_{0}a\beta^{3}\varepsilon_{1,2},\label{SysD}\\
\frac{d\varepsilon_{1,2}}{ds}  &  =(g-2a\delta_{1,2})\beta,\label{SysE}\\
\frac{d\kappa_{1,2}}{ds}  &  =g\delta_{1,2}-a\delta_{1,2}^{2}+c_{0}a\beta
^{2}\varepsilon_{1,2}^{2}. \label{SysF}%
\end{align}
Here,%
\begin{equation}
\alpha=\frac{1}{4a}\frac{\mu^{\prime}}{\mu}-\frac{d}{2a} \label{Alpha}%
\end{equation}
and solutions of the system (\ref{SysA})--(\ref{SysF}) are given by
(\ref{MKernel})--(\ref{FKernel}) and (\ref{MKernelOsc})--(\ref{FKernelOsc})
for $c_{0}=0$ and $c_{0}=1,$ respectively.
\end{lemma}

\begin{proof}
For a computer algebra derivation, see the \textsl{Mathematica}
notebook~\cite{notebook}, which is available as a supplementary material on
the article's website.
\end{proof}

In principle, our substitution (\ref{Ansatz2DParabolic}) can be thought of as
a generalized lens transformation in nonlinear paraxial optics
(cf.~\cite{KuzTur85}, \cite{Niederer72}, \cite{Niederer73}, \cite{Talanov70},
\cite{Tao09}, \cite{VlasovTalanovParabolicEq}). De facto, we have found a
\textquotedblleft proper\textquotedblright\ system of spatial coordinates
$(\xi,\eta,\tau)$ which automatically takes into consideration
\textquotedblleft imperfections\textquotedblright\ of initial data and turbid
medium in linear and quadratic approximations.

\textit{Note.} An algorithmic proof of the one-dimensional version of this lemma is
given in \cite{Koutschan11}.

\section{Multi-parameter Laser Beams and Their Special Cases}
\label{sec:multi}

With the help of the generalized lens transformation described in
Lemma~\ref{lem:nonlinear} and available explicit solutions from quantum
mechanics one can analyze, in a unified form, a large class of multi-parameter
modes for the corresponding linear parabolic wave equations in $1D$ and $2D$
weakly inhomogeneous media which are objects of interest in paraxial optics.
Some of these solutions have been already demonstrated in recent laser experiments but others,
which have quite exotic and spectacular properties according to our numerical simulations,
yet deserve an observation.

\subsection{Airy Beams}

\label{sec:AiryBeams}

In quantum mechanics, the time-dependent Schr\"{o}dinger equation for a free
particle (or the normalized paraxial wave equation in optics
\cite{Dodonov:Man'koFIAN87}, \cite{SiviloglouChris07} also known as the
parabolic equation \cite{Fock65}, \cite{VlasovTalanovParabolicEq}),%
\begin{equation}
i\psi_{t}+\psi_{xx}=0, \label{FreeSchroedinger}%
\end{equation}
by the following ansatz%
\begin{equation}
\psi(x,t) =e^{i(x-2t^{2}/3)t}\ F(x-t^{2}) \label{FreeSchroedingerSubstitution}%
\end{equation}
can be reduced to the Airy equation:%
\begin{equation}
F^{\prime\prime}=zF,\quad z=x-t^{2}, \label{AiryEquation}%
\end{equation}
whose bounded solutions are the Airy functions $F=k\Ai(z)$ (up
to a multiplicative constant~$k$) with well-known asymptotics as
$z\rightarrow\pm\infty$ \cite{Fock65}, \cite{Olver}.

The nonspreading Airy beams, which accelerate without any external force, were
introduced by Berry and Balazs \cite{BerryBalazs79} (see also \cite{Besetal94}%
, \cite{DanilovKuznetsovSmorodinskii80}, \cite{Greenberg80}, and
\cite{UnnikishnanRau96} for further exploration of different aspects of this
result). These nonspreading and freely accelerating wave packets have been
demonstrated in both one- and two-dimensional configurations as
quasi-diffraction-free optical beams \cite{SiviloglouChris07},
\cite{Siviloglouetal07} thus generating a considerable interest to this
phenomenon (see \cite{Abd10}, \cite{AbramRaz11}, \cite{AmentPolynkinMoloney11}%
, \cite{Bandres09}, \cite{BaddresVega07}, \cite{Bandresetal13},
\cite{BekSeg11}, \cite{BesierisShaarawi07}, \cite{Chenetal10},
\cite{Chenetal11}, \cite{Chongetal10}, \cite{Davisetal08}, \cite{Davisetal09},
\cite{Deng11}, \cite{KamSegChris11}, \cite{KasparianWolf09},
\cite{Lottietal11}, \cite{PangGburVisser11}, \cite{Polyn09}, \cite{RudMar11},
\cite{Torre13} and the references therein).

Equation \eqref{FreeSchroedinger} possesses a nontrivial symmetry
\cite{Niederer72}:%
\begin{equation}
i\psi_{t}+\psi_{xx}=0\quad\rightarrow\quad i\chi_{\tau}+\chi_{\xi\xi}=0,
\label{SchroedingerGroupEquation}%
\end{equation}
under the following transformation:%
\begin{align}
&  \psi(x,t)=\sqrt{\frac{\beta(0)}{1+4\alpha(0)t}}\ \exp i\left(  \frac
{\alpha(0)x^{2}+\delta(0)x-\delta^{2}(0)t}{1+4\alpha(0)t}+\kappa(0)\right)
\label{SchroedingerGroup}\\
&  \quad\quad\quad\times\chi\left(  \frac{\beta(0)x-2\beta(0)\delta
(0)t}{1+4\alpha(0)t}+\varepsilon(0),\ \frac{\beta^{2}(0)t}{1+4\alpha
(0)t}-\gamma(0)\right)  ,\nonumber
\end{align}
which is usually called the Schr\"{o}dinger group, and/or the maximum (known)
kinematical invariance group of the free Schr\"{o}dinger equation (see also
\cite{BandresGuiz09}, \cite{BoySharpWint},
\cite{DanilovKuznetsovSmorodinskii80}, \cite{Lop:Sus:VegaGroup},
\cite{LopSusVegaHarm}, \cite{Miller77}, \cite{Niederer73}, \cite{Torre13} and
the references therein; the subgroups and their invariants are discussed in
\cite{BoySharpWint}, \cite{Mah:Sus12}; the group parameters $\alpha(0),$
$\beta(0),$ $\gamma(0)=0,$ $\delta(0),$ $\varepsilon(0),$ and $\kappa(0)=0$
are chosen as initial data of the corresponding Riccati-type system
\cite{Lop:Sus:VegaGroup}).

As a result, in paraxial optics, the multi-parameter Airy modes are given by%
\begin{align}
B(x,s)  &  =\sqrt{\frac{\beta(0)}{1+4\alpha(0)s}}\ \exp\left(  i\,\frac
{\alpha(0)x^{2}+\delta(0)x-\delta^{2}(0)s}{1+4\alpha(0)s}\right)
\label{MultiAiry}\\
&  \times\exp\left(  \frac{i\beta^{2}(0)s}{1+4\alpha(0)s}\left(
\varepsilon(0)+\frac{\beta(0)x-2\beta(0)\delta(0)s}{1+4\alpha(0)s}-\frac{2}%
{3}\frac{\beta^{4}(0)s^{2}}{(1+4\alpha(0)s)^{2}}\right)  \right) \nonumber\\
&  \quad\times\Ai\left(  \varepsilon(0)+\frac{\beta
(0)x-2\beta(0)\delta(0)s}{1+4\alpha(0)s}-\frac{\beta^{4}(0)s^{2}}%
{(1+4\alpha(0)s)^{2}}\right) \nonumber
\end{align}
as particular solutions of the parabolic equation $iB_{s}%
+B_{xx}=0.$ (One can choose $\gamma(0)=$ $\kappa(0)=0$ in the explicit action
(\ref{SchroedingerGroup}) of the Schr\"{o}dinger group without loss of
generality.) The nonspreading case of Berry and Balazs \cite{BerryBalazs79}
occurs when $\alpha(0)=0$ in our notation. Other important special cases are
discussed in \cite{Mah:Sus12}, \cite{SiviloglouChris07},
\cite{Siviloglouetal07} (see also the references therein). The direct
verification by substitution and a computer algebra derivation of the
parabolic equation for the beams~\eqref{MultiAiry} is given
in~\cite{notebook} (see Section~\ref{sec:CA} for more details). Although
nowadays Airy and related beams are well documented \cite{Abd10},
\cite{Bandresetal13}, \cite{BesierisShaarawi14}, \cite{Chongetal10}, \cite{Davisetal08},
\cite{Davisetal09}, \cite{Lottietal11}, \cite{SiviloglouChris07},
\cite{Siviloglouetal07}, \cite{Torre10},
Figure~\ref{fig.Airy} represents an example of configuration which
yet deserves the experimental observation. Our
solution resembles, in the linear approximation, main features of rogue waves
\cite{Khetal09}, \cite{Mah:Sus12}, \cite{Smith76} 
(a simple animation is given in~\cite{notebook}).

\begin{figure}[htbp]
\centering \includegraphics[width=0.7\textwidth]{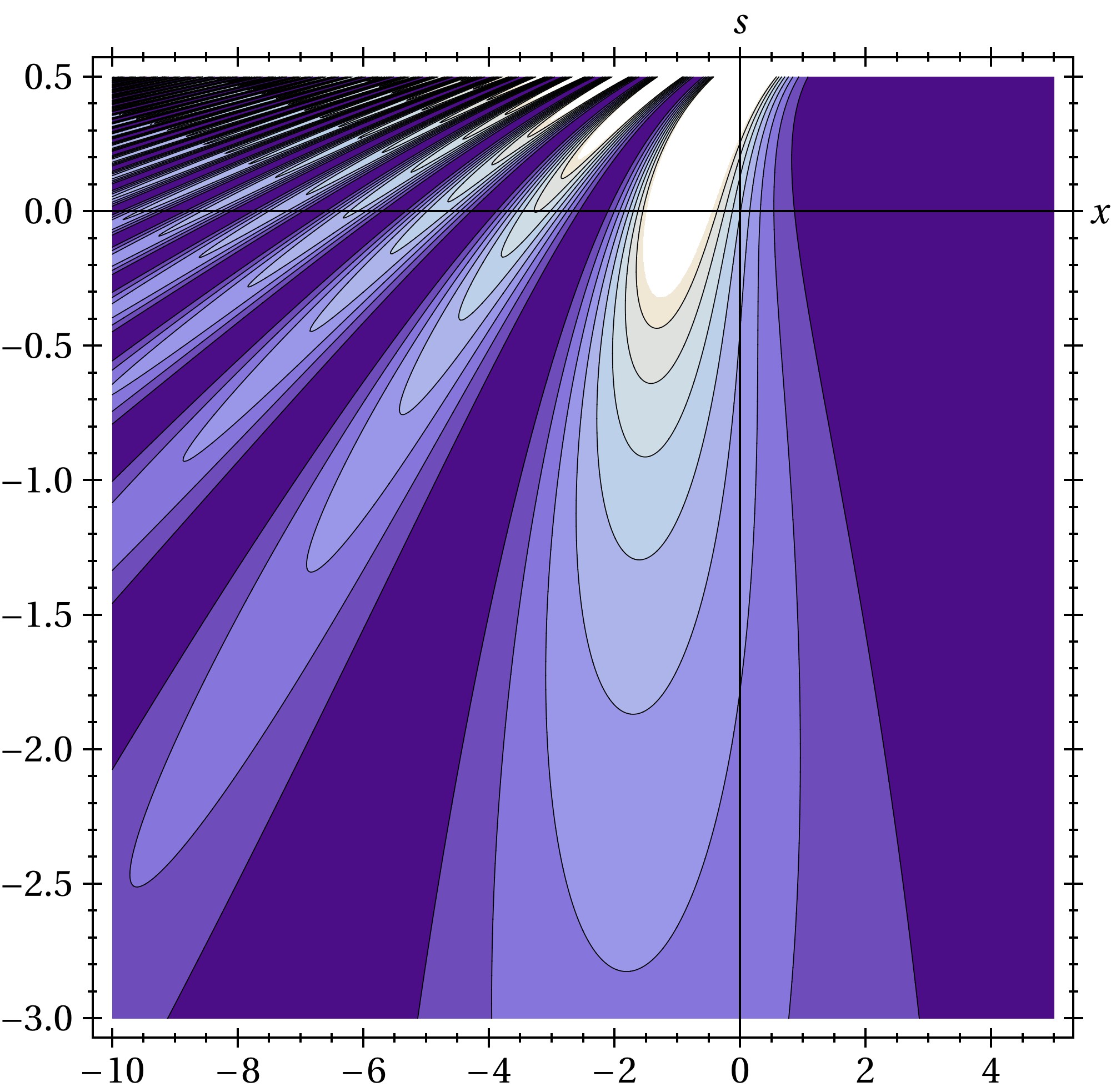}
\caption{Levels of intensity of the Airy mode \eqref{MultiAiry}
when $\alpha(0)=-1/4,$ $\beta(0)=1,$ and $\delta(0)=\varepsilon(0)=0.$}
\label{fig.Airy}
\end{figure}

\subsection{Oscillating and Breathing Her\-mite-Gaussian Beams}

For a $1D$ inhomogeneous paraxial wave equation with quadratic refractive
index (a lens-like medium \cite{Kogelnik65}, \cite{YarivYeh}),%
\begin{equation}
2iA_{s}+A_{xx}-x^{2}A=0, \label{Schroudinger}%
\end{equation}
an important multi-parameter family of particular solutions can be presented
as follows \cite{KrySusVegaMinimum}, \cite{Mah:Sua:Sus13}:%
\begin{equation}
\ A_{n}(x,s)=e^{i\left(  \alpha x^{2}+\delta x+\kappa\right)  +i\left(
2n+1\right)  \gamma}\sqrt{\frac{\beta}{2^{n}n!\sqrt{\pi}}}e^{-\left(  \beta
x+\varepsilon\right)  ^{2}/2}\ H_{n}(\beta x+\varepsilon), \label{AFunction}%
\end{equation}
where $H_{n}(x)$ are the Hermite polynomials \cite{Ni:Su:Uv} and%
\begin{align}
\alpha(s)  &  =\frac{\alpha_{0}\cos2s+\sin2s\left(  \beta_{0}^{4}+4\alpha
_{0}^{2}-1\right)  /4}{\beta_{0}^{4}\sin^{2}s+\left(  2\alpha_{0}\sin s+\cos
s\right)  ^{2}},\label{hhA}\\
\beta(s)  &  =\frac{\beta_{0}}{\sqrt{\beta_{0}^{4}\sin^{2}s+\left(
2\alpha_{0}\sin s+\cos s\right)  ^{2}}},\label{hhB}\\
\gamma(s)  &  =-\frac{1}{2}\arctan\frac{\beta_{0}^{2}\tan s}{1+2\alpha_{0}\tan
s},\label{hhG}\\
\delta(s)  &  =\frac{\delta_{0}\left(  2\alpha_{0}\sin s+\cos s\right)
+\varepsilon_{0}\beta_{0}^{3}\sin s}{\beta_{0}^{4}\sin^{2}s+\left(
2\alpha_{0}\sin s+\cos s\right)  ^{2}},\label{hhD}\\
\varepsilon(s)  &  =\frac{\varepsilon_{0}\left(  2\alpha_{0}\sin s+\cos
s\right)  -\beta_{0}\delta_{0}\sin s}{\sqrt{\beta_{0}^{4}\sin^{2}s+\left(
2\alpha_{0}\sin s+\cos s\right)  ^{2}}},\label{hhE}\\
\kappa(s)  &  =\sin^{2}s\frac{\varepsilon_{0}\beta_{0}^{2}\left(  \alpha
_{0}\varepsilon_{0}-\beta_{0}\delta_{0}\right)  -\alpha_{0}\delta_{0}^{2}%
}{\beta_{0}^{4}\sin^{2}s+\left(  2\alpha_{0}\sin s+\cos s\right)  ^{2}%
}\label{hhK}\\
&  \quad+\frac{1}{4}\sin2s\frac{\varepsilon_{0}^{2}\beta_{0}^{2}-\delta
_{0}^{2}}{\beta_{0}^{4}\sin^{2}s+\left(  2\alpha_{0}\sin s+\cos s\right)
^{2}}.\nonumber
\end{align}
The real or complex-valued parameters $\alpha_{0},$ $\beta_{0}\neq0,$
$\gamma_{0}=0,$ $\delta_{0},$ $\varepsilon_{0},$ $\kappa_{0}=0$ are initial
data of the corresponding Ermakov-type system \cite{Lan:Lop:Sus},
\cite{Lop:Sus:VegaGroup}.\footnote{From now on, we abbreviate $\alpha
_{0}=\alpha(0),$ etc for the sake of compactness.} A direct
\textsl{Mathematica} verification can be found in~\cite{notebook}. (A similar
harmonic motion of cold trapped atoms was experimentally realized
\cite{LeibfriedetalWineland03}.)

These \textquotedblleft missing\textquotedblright\ solutions that are omitted
in all textbooks on quantum mechanics (see \cite{LopSusVegaHarm} and
\cite{Marhic78}) provide a new multi-parameter family of oscillating
Hermite-Gaussian beams in parabolic (self-focusing fiber) waveguides, which
deserve an experimental observation; special cases were theo\-reti\-cally
studied earlier in \cite{Agrawaletal74}, \cite{Eichmann71}, \cite{Fock65},
\cite{GoncharenkoBook}, \cite{Kogelnik65}, \cite{Yariv3rdEd},
\cite{VinRudSuxBook79}. For graphical examples see Figures 1 and 2 of
Ref.~\cite{Mah:Sua:Sus13}. These modes are orthonormal for real-valued
parameters. As a result, every $L^{2}$ distribution of monochromatic light can
be expanded in terms of these modes. The corresponding generalized coherent or
minimum-uncertainty squeezed states are analyzed in \cite{KrySusVegaMinimum}.

\subsection{Her\-mite-Gaussian Beams}
\label{sec:HG}

The homogeneous paraxial wave equation,%
\begin{equation}
2iB_{s}+B_{xx}=0, \label{HomParWaveEq}%
\end{equation}
can be transformed by the substitution,%
\begin{equation}
B(x,s)=\frac{1}{\left(  1+s^{2}\right)  ^{1/4}}\exp\left(  \frac{isx^{2}%
}{2\left(  1+s^{2}\right)  }\right)  \ A\left(  \frac{x}{\sqrt{1+s^{2}}%
},\ \arctan s\right)  , \label{LinQuadTransform}%
\end{equation}
into the inhomogeneous one (\ref{Schroudinger}) (see \cite{Lop:Sus:VegaGroup}
and the references therein; a \textsl{Mathematica} verification can be found
in~\cite{notebook}). Composition of (\ref{WaveFunction}) and
(\ref{LinQuadTransform}) results in the following multi-parameter family of
\textquotedblleft spreading\textquotedblright\ solutions to the parabolic
equation \eqref{HomParWaveEq}:
\begin{align}
&  B_{n}(x,s)=\sqrt{\frac{\beta_{0}}{2^{n}n!\sqrt{\pi\left(  \left(
1+2\alpha_{0}s\right)  ^{2}+\beta_{0}^{4}s^{2}\right)  }}}\label{1DLinNewSols}%
\\
&  \quad\times\exp\left(  -\frac{\left(  \beta_{0}x+\varepsilon_{0}\right)
^{2}+2s\left(  \alpha_{0}\varepsilon_{0}-\delta_{0}\beta_{0}\right)
\varepsilon_{0}-i\left(  2x\left(  \alpha_{0}x+\delta_{0}\right)  -s\delta
_{0}^{2}\right)  }{2\left(  1+2\alpha_{0}s+i\beta_{0}^{2}s\right)  }\right)
\nonumber\\
&  \quad\times\exp\left(  -i\left(  n+\frac{1}{2}\right)  \arctan\left(
\frac{\beta_{0}^{2}s}{1+2\alpha_{0}s}\right)  \right)  H_{n}\left(
\frac{\beta_{0}\left(  x-\delta_{0}s\right)  +\left(  1+2\alpha_{0}s\right)
\varepsilon_{0}}{\sqrt{\left(  1+2\alpha_{0}s\right)  ^{2}+\beta_{0}^{4}s^{2}%
}}\right) \nonumber
\end{align}
for real or complex initial data \cite{Mah:Sua:Sus13}. The direct derivation
is also provided in~\cite{notebook}. It is worth noting that both of our
parameters $\varepsilon_{0}\neq0$ (shift) and $\delta_{0}\neq0$ (phase)
describe, in a natural way, the beam deviation from the optical axis and a
successive oblique propagation in an optical system, which is not usually
discussed in detail in the literature. (This solution is also relevant to the
concept of paraxial group \cite{BandresGuiz09} that comprises $2D$
transformations of a beam propagating through misaligned (tilted, translated,
or rotated) $ABCD$ optical systems.)

\textit{Note.} When $n=0,$ the intensity distribution $\left\vert
B_{0}(x,s)\right\vert ^{2}$ is normal in every beam cross section and the
width of that Gaussian intensity profile changes along the $s$ axis. The beam
waist/focal point, when $\nabla\left\vert B_{0}(x,s)\right\vert ^{2}=0$ and $\max
\left\vert B_{0}(x,s)\right\vert ^{2}=\sqrt{4\alpha_{0}^{2}+\beta_{0}^{4}%
}\ /\left\vert \beta_{0}\right\vert \!\sqrt{\pi},$ occurs at%
\[
x_{0}=-\frac{2\alpha_{0}\delta_{0}+\beta_{0}^{3}\varepsilon_{0}}{4\alpha
_{0}^{2}+\beta_{0}^{4}},\qquad s_{0}=-\frac{2\alpha_{0}}{4\alpha_{0}^{2}%
+\beta_{0}^{4}}.
\]
In the limit $\beta_{0}\rightarrow0,$ we obtain the traditional definition of
focus in the thin lens approximation \cite{AhmanNikBook},
\cite{VinRudSuxBook79}:%
\[
\left.  \left(  B_{n}(x,0)/\sqrt{\beta_{0}}\right)  \right\vert _{\beta
_{0}\rightarrow0}=Ce^{i\delta_{0}x}\ e^{-ix^{2}/\left(  2s_{0}\right)  }.
\]
The beam radius, related to standard deviation, is defined as the distance at
which the amplitude is $1/e$ times of that on the mean \cite{KogelnikLi66}.
Thus the smallest radius is observed at the focal point:%
\[
r_{0}=\frac{\left\vert \beta_{0}\right\vert }{\sqrt{4\alpha_{0}^{2}+\beta
_{0}^{4}}}.
\]
(Details are given in \cite{notebook} together with a graphical example of
\ \textquotedblleft self-focusing\textquotedblright\ of the corresponding
Gaussian mode \textquotedblleft without any external force\textquotedblright.)

Among various special cases of these multi-parameter solutions are the
so-called elegant Hermite-Gaussian beams. In our notation, they occur for the
complex-valued parameters when $4\alpha_{0}^{2}+\beta_{0}^{4}=0.$\ The
substitution%
\begin{equation}
\frac{1+2\alpha_{0}s+i\beta_{0}^{2}s}{\sqrt{\left(  1+2\alpha_{0}s\right)
^{2}+\beta_{0}^{4}s^{2}}}=\exp\left(  i\arctan\left(  \frac{\beta_{0}^{2}%
s}{1+2\alpha_{0}s}\right)  \right)  \label{arctansub}%
\end{equation}
followed by $2\alpha_{0}=i\beta_{0}^{2}$ results in%
\begin{align}
&  B_{n}^{\left(  \text{el}\right)  }(x,s)=\sqrt{\frac{\beta_{0}}%
{2^{n}n!\left(  1+2i\beta_{0}^{2}s\right)  ^{n+1}\sqrt{\pi}}}\ H_{n}\left(
\frac{\beta_{0}\left(  x-\delta_{0}s+i\beta_{0}\varepsilon_{0}s\right)
+\varepsilon_{0}}{\sqrt{1+2i\beta_{0}^{2}s}}\right)
\label{ElegantGaussHermite}\\
&  \quad\qquad\qquad\times\exp\left(  -\frac{2\beta_{0}^{2}x^{2}+\left(
2\beta_{0}\varepsilon_{0}-i\delta_{0}\right)  \left(  2x-\delta_{0}s\right)
-\left(  1+2i\beta_{0}^{2}s\right)  \varepsilon_{0}^{2}}{2\left(
1+2i\beta_{0}^{2}s\right)  }\right)  .\nonumber
\end{align}
When $n=0,$ one gets the multi-parameter fundamental Gaussian modes. In this
case,%
\begin{align}
&  \left\vert B_{0}^{\left(  \text{el}\right)  }(x,s)\right\vert ^{2}%
=\frac{\beta_{0}}{\sqrt{\pi\left(  1+4\beta_{0}^{4}s^{2}\right)  }}\exp\left(
-\frac{2\beta_{0}^{2}\left(  x-\delta_{0}s\right)  ^{2}+2\beta_{0}%
\varepsilon_{0}\left(  x-\delta_{0}s\right)  +\varepsilon_{0}^{2}\left(
1+2\beta_{0}^{4}s^{2}\right)  }{1+4\beta_{0}^{4}s^{2}}\right)
,\label{FundamentalElegant}\\
&  \qquad\qquad\qquad\qquad\int_{-\infty}^{\infty}\left\vert B_{0}^{\left(
\text{el}\right)  }(x,s)\right\vert ^{2}\ dx=\frac{e^{-\varepsilon_{0}^{2}/2}%
}{\sqrt{2}}.\nonumber
\end{align}
These optical fields obey a certain \textquotedblleft
propagation-invariant\ similarity rule\textquotedblright:
\[
\left\vert B_{0}^{\left(  \text{el}\right)  }(x,s)\right\vert ^{2}=\frac
{\beta_{0}e^{-k^{2}/2}}{\sqrt{\pi\left(  1+4\beta_{0}^{4}s^{2}\right)  }%
},\qquad k=\text{constant},%
\]
provided that $2\beta_{0}\left(  x-\delta_{0}s\right)  =-\varepsilon_{0}%
\pm\sqrt{\left(  k^{2}-\varepsilon_{0}^{2}\right)  \left(  1+4\beta_{0}%
^{4}s^{2}\right)  }$ and $k^{2}\geq\varepsilon_{0}^{2}.$ Thus, our solution
describes an \textquotedblleft oblique propagation\textquotedblright\ of the
laser beam with respect to the optical axis (approaching the corresponding
slanted asymptotes as $s\rightarrow\infty).$ For instance, the best
confinement of optical energy occurs around the line $x=\delta_{0}s,$ which
becomes the direction of the beam propagation, when $\varepsilon_{0}=0.$ This
simple example shows how one can use our extra parameters in order to aim the
laser beam and to maximize its intensity. A graphical example is provided in
Figure~\ref{fig.slanted}; see also \cite{notebook} for more details.

\begin{figure}[htbp]
\centering \includegraphics[width=0.65\textwidth]{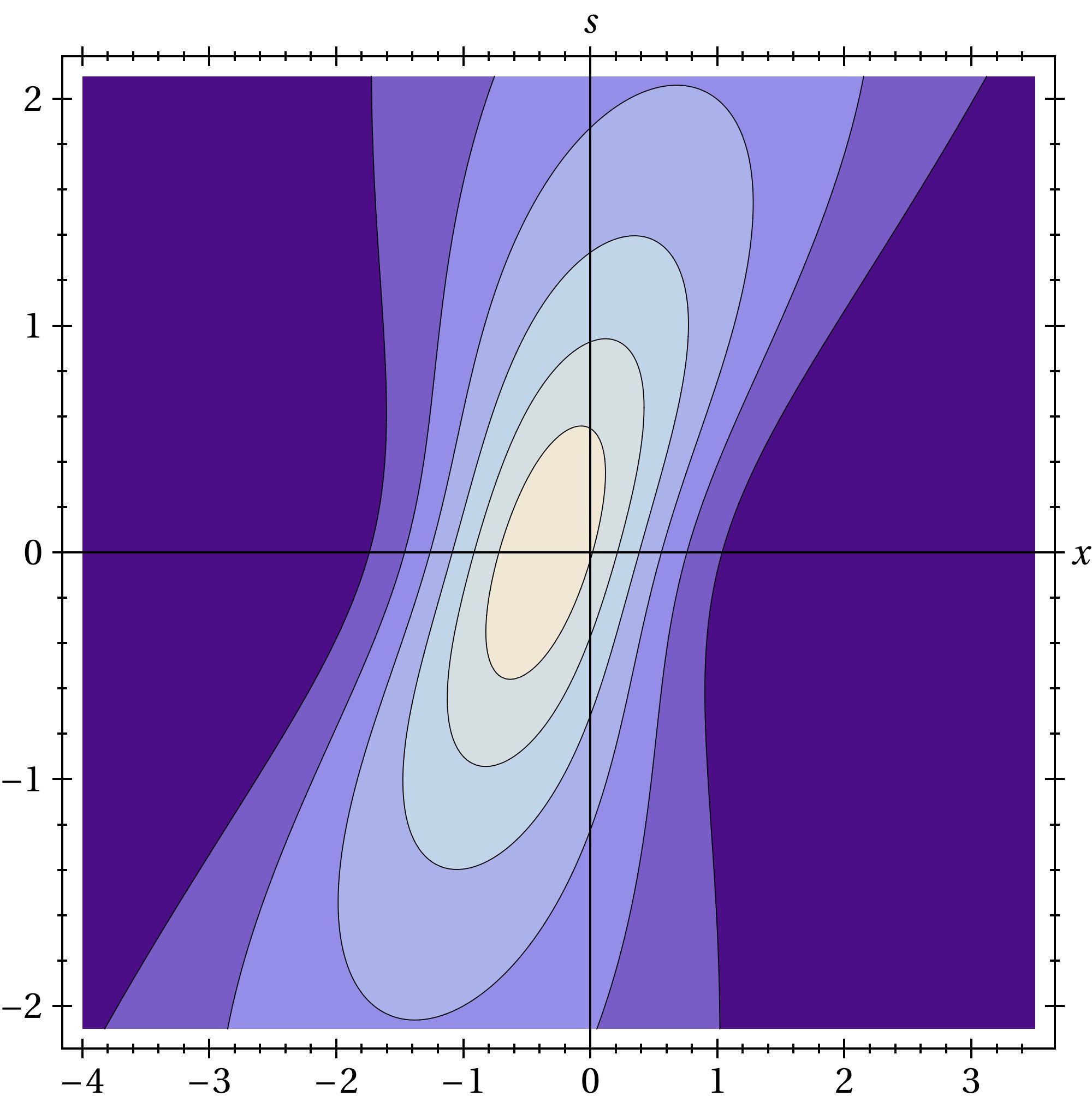}
\caption{Levels of intensity of the slanted beam \eqref{FundamentalElegant}
when $\beta_0=2^{-1/2}$ and $\delta_0=\varepsilon_0=1/2.$}
\label{fig.slanted}
\end{figure}

Moreover, by the expansion transformation of the Schr\"{o}dinger group
\cite{Lop:Sus:VegaGroup}:%
\[
B(x,s)=\frac{1}{\sqrt{1+ms}}\exp\left(  \frac{imx^{2}}{2\left(  1+ms\right)
}\right)  \ C\left(  \frac{x}{1+ms},\ \frac{s}{1+ms}\right)  \qquad\left(
m=\text{constant}\right)
\]
one arrives at the following Gaussian package:%
\begin{align*}
&  B_{0}^{\left(  \text{el,exp}\right)  }(x,s)=\exp\left(  \frac{imx^{2}%
}{2\left(  1+ms\right)  }\right)  \sqrt{\frac{\beta_{0}}{\left(
1+ms+2i\beta_{0}^{2}s\right)  \sqrt{\pi}}}\\
&  \times\exp\left(  -\frac{2\beta_{0}^{2}x^{2}+2\left(  1+ms\right)  \left(
\beta_{0}\varepsilon_{0}-i\delta_{0}\right)  x+\left(  1+ms\right)  \left(
\varepsilon_{0}^{2}+s\left(  \left(  \varepsilon_{0}^{2}(m+2i\beta_{0}%
^{2}\right)  -2\beta_{0}\delta_{0}\varepsilon_{0}+i\delta_{0}^{2}\right)
\right)  }{2\left(  1+ms\right)  \left(  1+ms+2i\beta_{0}^{2}s\right)
}\right)
\end{align*}
$\allowbreak$(see also \cite{notebook} for a direct verification). One may say
that its spatial evolution also resembles the generation of an optical
\textquotedblleft rogue wave\textquotedblright\ which is appearing at a
certain point and then dissipating. A graphical example of the intensity and
optical energy localization is provided in \cite{notebook}.

Special families of Gaussian beams have found significant applications in
science, biomedicine, and technology. Among them, the fundamental Gaussian
mode described by Eq.~(\ref{ElegantGaussHermite}), when $n=0$ and $\alpha
_{0}=\delta_{0}=\varepsilon_{0}=0,$ is the most useful one. According to
\cite{AbramAlievaRodrigo12}, the laser beams of this kind are utilized for the
material cutting and surgery, for data reading in CD-DVD players and in
optical remote sensing technology, and for microparticle trapping and atom
cooling. Thus, telecommunication networks including the internet are based
upon optical waveguide systems in which fundamental Gaussian modes are
propagated in a wavelength multiplexing configuration.

In general, our multi-parameter solutions (\ref{1DLinNewSols}) can be thought
of as the Hermite-Gaussian beams with \textquotedblleft aberration/astigmatic
elements\textquotedblright\ that are useful, for example,\ in the propagation
of paraxial beams through a misaligned optical system (see
Refs.~\cite{AbramAlievaRodrigo12}, \cite{AbramVol91}, \cite{AbramVolBook},
\cite{AhmanNikBook}, \cite{BandresGuiz09}. \cite{KogelnikLi66},
\cite{PratesiRonchi77}, \cite{Sieg73}, \cite{YarivYeh}, \cite{VinRudSuxBook79}%
, \cite{Wunsche89} for further examples of these important modes in one and two-dimensions).

\textit{Note.} Although the multi-parameter elegant Hermite-Gaussian beams are
not orthogonal, the corresponding integral:%
\[
\int_{-\infty}^{\infty}\left(  B_{n}^{\left(  \text{el}\right)  }\left(
x,s\right)  \right)  ^{\ast}B_{m}^{\left(  \text{el}\right)  }\left(
x,s\right)  \ dx
\]
can be evaluated in terms of generalized hypergeometric functions in a way
that is similar to \cite{KrySusVegaMinimum}. An investigation of certain
minimization properties may be of interest.

\subsection{Breathing Spiral Laser Beams}

\label{sec:BreathingSpiral}

By the ansatz $\Psi(x,y,t) =\chi(\xi,\eta,\tau) ,$ $T=-\tau$ and
\begin{equation}
\left(
\begin{array}
[c]{c}%
X\\
Y
\end{array}
\right)  =\left(
\begin{array}
[c]{cc}%
\cos\omega\tau & -\sin\omega\tau\\
\sin\omega\tau & \cos\omega\tau
\end{array}
\right)  \left(
\begin{array}
[c]{c}%
\xi\\
\eta
\end{array}
\right)  \label{Rotation}%
\end{equation}
($\omega=\ $constant), Equation \eqref{2DSchroedingerTransformed} with
$c_{0}=1$ can be transformed to the equation of motion for the isotropic
planar harmonic oscillator in a perpendicular uniform magnetic field:
\begin{equation}
i\Psi_{T}+\Psi_{XX}+\Psi_{YY}=\left(  X^{2}+Y^{2}\right)  \Psi+i\omega\left(
X\Psi_{Y}-Y\Psi_{X}\right)  . \label{2DSchroedingerReduced}%
\end{equation}
The latter equation was solved in the early days of quantum mechanics by Fock
\cite{Fock28-1} in polar coordinates, $X=R\cos\Theta$ and $Y=R\sin\Theta:$%
\begin{align}
&  \Psi(R,\Theta,T) =\sqrt{\frac{n!}{\pi\left(  n+\left\vert m\right\vert
\right)  !}}\ e^{-iET}e^{im\Theta}R^{\left\vert m\right\vert }e^{-R^{2}%
/2}L_{n}^{\left\vert m\right\vert }\left(  R^{2}\right)
,\label{FockLaguerre2D}\\
&  \qquad\qquad\qquad\quad E=4n+2\left(  \left\vert m\right\vert +1\right)
-m\omega\nonumber
\end{align}
$(m=\pm0,\pm1,\ldots,n=0,1,\ldots) $ in terms of Laguerre polynomials
\cite{Ni:Su:Uv}. This wave function coincides, up to a simple factor, with the
one for a flat isotropic oscillator without magnetic field. Therefore, its
development in terms of (\ref{2DGaussHermite}) for standard harmonics is a
$2D$ special case of the multi-dimensional expansions from \cite{Ni:Su:Uv}
(see also \cite{Cor-Sot:Sus}, \cite{MeCoSus} and the references therein).

By back substitution, one arrives at a general family of spiral solutions in
inhomogeneous media. For example, the $2D$ paraxial wave equation%
\begin{equation}
2iA_{s}+A_{xx}+A_{yy}=\left(  x^{2}+y^{2}\right)  A \label{2DStandardForm}%
\end{equation}
possesses the following Laguerre-Gaussian modes \cite{Mah:Sua:Sus13}
\begin{align}
A_{n}^{m}(x,y,s)  &  =\beta\sqrt{\frac{n!}{\pi\left(  n+m\right)  !}%
}e^{i\left(  \alpha\left(  x^{2}+y^{2}\right)  +\delta_{1}x+\delta_{2}%
y+\kappa_{1}+\kappa_{2}\right)  }e^{i\left(  2n+m+1\right)  \gamma}%
\big(\beta(x\pm iy)+\varepsilon_{1}\pm i\varepsilon_{2}\big)^{m}%
\label{2DGaussLaguerre}\\
&  \qquad\times e^{-(\beta x+\varepsilon_{1})^{2}/2-(\beta y+\varepsilon
_{2})^{2}/2}L_{n}^{m}\left(  (\beta x+\varepsilon_{1})^{2}+(\beta
y+\varepsilon_{2})^{2}\right)  ,\quad m\geq0\nonumber
\end{align}
(by the action of Schr\"{o}dinger's group; see \cite{Lop:Sus:VegaGroup},
\cite{LopSusVegaHarm}, \cite{Mah:Sus13} and the references therein for
classical accounts). Here, Equations \eqref{hhA}--\eqref{hhK} are utilized for
complex or real-valued parameters $\alpha_{0},$ $\beta_{0}\neq0,$ $\delta
_{0}^{(1,2)},$ $\varepsilon_{0}^{(1,2)}$ (the last two sets may be different
for $x$ and $y$ variables, respectively). Examples are shown in Figures 3 and
4 of Ref.~\cite{Mah:Sua:Sus13}.

In addition, a special Gaussian form of our solution (\ref{2DGaussHermite})
gives a general example of spiral elliptic beams discussed in
\cite{GoncharenkoBook}, \cite{YarivYeh}.

\subsection{Laguerre-Gaussian Beams}

\label{sec:LaguerreGaussian}

The homogeneous parabolic equation,%
\begin{equation}
2iB_{s}+B_{xx}+B_{yy}=0, \label{2DFree}%
\end{equation}
with the help of\footnote{Both Equations, (\ref{2DStandardForm}) and
(\ref{2DFree}), are obviously invariant under plane rotations.}%
\begin{equation}
B(x,y,s)=\frac{1}{\left(  s^{2}+1\right)  ^{1/2}}\exp\left(  \frac{is\left(
x^{2}+y^{2}\right)  }{2\left(  s^{2}+1\right)  }\right)  \ A\left(  \frac
{x}{\sqrt{s^{2}+1}},\frac{y}{\sqrt{s^{2}+1}},\ \arctan s\right)
\label{2DLinQuad}%
\end{equation}
can be reduced to the standard form (\ref{2DStandardForm}). A multi-parameter
solution is given by \cite{Mah:Sua:Sus13}%
\begin{align}
&  B_{n}^{m}(x,y,s)=\frac{1}{1+2\alpha_{0}s+i\beta_{0}^{2}s}\exp\left(
-i\left(  m+2n\right)  \arctan\left(  \dfrac{s\beta_{0}^{2}}{1+2\alpha_{0}%
s}\right)  \right) \label{2DGaussLaguerreSpreading}\\
&  \quad\times\exp\left(  -is\frac{\left.  \delta_{0}^{(1)}\right.
^{2}+\left.  \delta_{0}^{(2)}\right.  ^{2}}{2\left(  1+2\alpha_{0}s+i\beta
_{0}^{2}s\right)  }\right) \nonumber\\
&  \quad\times\exp\left(  \frac{\left(  2i\alpha_{0}-\beta_{0}^{2}\right)
\left(  x^{2}+y^{2}\right)  -2\left(  \beta_{0}\varepsilon_{0}^{\left(
1\right)  }-i\delta_{0}^{(1)}\right)  x-2\left(  \beta_{0}\varepsilon
_{0}^{(2)}-i\delta_{0}^{(2)}\right)  y}{2\left(  1+2\alpha_{0}s+i\beta_{0}%
^{2}s\right)  }\right) \nonumber\\
&  \quad\times\exp\left(  \frac{2\beta_{0}s\left(  \delta_{0}^{(1)}%
\varepsilon_{0}^{(1)}+\delta_{0}^{(2)}\varepsilon_{0}^{(2)}\right)  -\left(
1+2\alpha_{0}s\right)  \left(  \left.  \varepsilon_{0}^{(1)}\right.
^{2}+\left.  \varepsilon_{0}^{(2)}\right.  ^{2}\right)  }{2\left(
1+2\alpha_{0}s+i\beta_{0}^{2}s\right)  }\right) \nonumber\\
&  \quad\times\left(  \frac{\beta_{0}(x+iy)-\left(  \delta_{0}^{(1)}%
+i\delta_{0}^{(2)}\right)  s+\left(  \varepsilon_{0}^{(1)}+i\varepsilon
_{0}^{(2)}\right)  \left(  1+2\alpha_{0}s\right)  }{\sqrt{\left(
1+2\alpha_{0}s\right)  ^{2}+\beta_{0}^{4}s^{2}}}\right)  ^{m}%
\!\!\!\!\nonumber\\
&  \quad\times L_{n}^{m}\left(  \frac{\left(  \beta_{0}\left(  x-\delta
_{0}^{(1)}s\right)  +\varepsilon_{0}^{(1)}\left(  1+2\alpha_{0}s\right)
\right)  ^{2}+\left(  \beta_{0}\left(  y-\delta_{0}^{(2)}s\right)
+\varepsilon_{0}^{(2)}\left(  1+2\alpha_{0}s\right)  \right)  ^{2}}{\left(
1+2\alpha_{0}s\right)  ^{2}+\beta_{0}^{4}s^{2}}\right) \nonumber
\end{align}
by the action of Schr\"{o}dinger's group. (The corresponding parameters are
initial data of the Ermakov-type system (\ref{SysA})--(\ref{SysF}); see
Lemma~\ref{lem:nonlinear}.)

\textit{Note.} An example of \textquotedblleft self-focusing\textquotedblright%
\ Gaussian mode, when $n=m=0,$ is presented in~\cite{notebook}. The
corresponding focal point, when $\max\left\vert B_{0}^{0}(x,y,s)\right\vert
^{2}=1+4\alpha_{0}^{2}/\beta_{0}^{4},$ is located at%
\[
x_{0}=-\frac{2\alpha_{0}\delta_{0}^{(1)}+\beta_{0}^{3}\varepsilon_{0}^{(1)}%
}{4\alpha_{0}^{2}+\beta_{0}^{4}},\quad y_{0}=-\frac{2\alpha_{0}\delta
_{0}^{(2)}+\beta_{0}^{3}\varepsilon_{0}^{(2)}}{4\alpha_{0}^{2}+\beta_{0}^{4}%
},\quad s_{0}=-\frac{2\alpha_{0}}{4\alpha_{0}^{2}+\beta_{0}^{4}}.
\]
It is worth noting that this mode describes the well-known effect of focusing
of a laser beam in a uniform medium after passing the lens/quadratic medium.
(In our approach, the quadratic, or lens-like, medium creates the
corresponding initial data for the focusing beam, in a mathematically natural way.)

For the set of complex-valued parameters, two special cases are of interest,
namely the multi-parameter \textquotedblleft elegant\textquotedblright%
\ Laguerre-Gaussian beams, when $2\alpha_{0}=i\beta_{0}^{2}:$%
\begin{align}
&  \left.  B_{n}^{m}(x,y,s)\right.  ^{\left(  \text{el}\right)  }=\left(
1+2i\beta_{0}^{2}s\right)  ^{-m-n-1}\exp\left(  -is\frac{\left.  \delta
_{0}^{(1)}\right.  ^{2}+\left.  \delta_{0}^{\left(  2\right)  }\right.  ^{2}%
}{2\left(  1+2i\beta_{0}^{2}s\right)  }\right) \label{ElegantLG}\\
&  \quad\times\exp\left(  -\frac{\beta_{0}^{2}\left(  x^{2}+y^{2}\right)
+\left(  \beta_{0}\varepsilon_{0}^{(1)}-i\delta_{0}^{(1)}\right)  x+\left(
\beta_{0}\varepsilon_{0}^{(2)}-i\delta_{0}^{(2)}\right)  y}{\left(
1+2i\beta_{0}^{2}s\right)  }\right) \nonumber\\
&  \quad\times\exp\left(  \frac{2\beta_{0}s\left(  \delta_{0}^{(1)}%
\varepsilon_{0}^{(1)}+\delta_{0}^{(2)}\varepsilon_{0}^{(2)}\right)  -\left(
1+i\beta_{0}^{2}s\right)  \left(  \left.  \varepsilon_{0}^{(1)}\right.
^{2}+\left.  \varepsilon_{0}^{(2)}\right.  ^{2}\right)  }{2\left(
1+2i\beta_{0}^{2}s\right)  }\right) \nonumber\\
&  \quad\times\left(  \beta_{0}\left(  x+iy\right)  -\left(  \delta_{0}%
^{(1)}+i\delta_{0}^{(2)}\right)  s+\left(  \varepsilon_{0}^{(1)}%
+i\varepsilon_{0}^{(2)}\right)  \left(  1+2\alpha_{0}s\right)  \right)
^{m}\!\!\!\!\nonumber\\
&  \quad\times L_{n}^{m}\left(  \frac{\left(  \beta_{0}\left(  x-\delta
_{0}^{(1)}s\right)  +\varepsilon_{0}^{(1)}\left(  1+i\beta_{0}^{2}s\right)
\right)  ^{2}+\left(  \beta_{0}\left(  y-\delta_{0}^{(2)}s\right)
+\varepsilon_{0}^{(2)}\left(  1+i\beta_{0}^{2}s\right)  \right)  ^{2}%
}{1+2i\beta_{0}^{2}s}\right)  ,\nonumber
\end{align}
and multi-parameter \textquotedblleft diffraction-free\textquotedblright%
\ Laguerre beams, when $2\alpha_{0}=-i\beta_{0}^{2}:$%
\begin{align}
&  \left.  B_{n}^{m}(x,y,s)\right.  ^{\left(  \text{dif}\right)  }=\left(
1-2i\beta_{0}^{2}s\right)  ^{n}\exp\left(  -\left(  \beta_{0}\varepsilon
_{0}^{(1)}-i\delta_{0}^{(1)}\right)  x-\left(  \beta_{0}\varepsilon_{0}%
^{(2)}-i\delta_{0}^{(2)}\right)  y\right) \label{DiffL}\\
&  \quad\times\exp\left(  \beta_{0}s\left(  \delta_{0}^{(1)}\varepsilon
_{0}^{(1)}+\delta_{0}^{(2)}\varepsilon_{0}^{(2)}\right)  -\frac{1-i\beta
_{0}^{2}s}{2}\left(  \left.  \varepsilon_{0}^{(1)}\right.  ^{2}+\left.
\varepsilon_{0}^{(2)}\right.  ^{2}\right)  -\frac{is}{2}\left(  \left.
\delta_{0}^{(1)}\right.  ^{2}+\left.  \delta_{0}^{(2)}\right.  ^{2}\right)
\right) \nonumber\\
&  \quad\times\left(  \beta_{0}(x+iy)-\left(  \delta_{0}^{(1)}+i\delta
_{0}^{(2)}\right)  s+\left(  \varepsilon_{0}^{(1)}+i\varepsilon_{0}%
^{(2)}\right)  \left(  1-i\beta_{0}^{2}s\right)  \right)  ^{m}%
\!\!\!\!\nonumber\\
&  \quad\times L_{n}^{m}\left(  \frac{\left(  \beta_{0}\left(  x-\delta
_{0}^{(1)}s\right)  +\varepsilon_{0}^{(1)}\left(  1+i\beta_{0}^{2}s\right)
\right)  ^{2}+\left(  \beta_{0}\left(  y-\delta_{0}^{(2)}s\right)
+\varepsilon_{0}^{(2)}\left(  1+i\beta_{0}^{2}s\right)  \right)  ^{2}%
}{1-2i\beta_{0}^{2}s}\right)  .\nonumber
\end{align}
For $m=n=0$ and $\varepsilon_{0}^{(1,2)}=0,$ this beam degenerates into the
ordinary plane wave propagating in the direction $\boldsymbol{r}=\left(
\delta_{0}^{(1)},\delta_{0}^{(2)},1\right)  .$
Polynomial solutions have also been discussed in \cite{AbramPol} and \cite{Dennis11}.

Among numerous special cases are the Laguerre-Gaussian beams discovered in
\cite{Belanger84}, \cite{PiesSCHechSha00}, \cite{VinRudSuxBook79},
\cite{Wunsche89}. By classical accounts \cite{AbramAlievaRodrigo12},
\cite{AbramVolBook}, \cite{GoncharenkoBook}, \cite{KogelnikLi66},
\cite{Smith00}, \cite{Yariv3rdEd}, \cite{YarivYeh} (see also the references
therein), the families of the Hermite-Gaussian and Laguerre-Gaussian modes
arise naturally as approximate eigenfunctions of the resonators with
rectangular or circular spherical/flat mirrors, respectively. They also serve
as models for eigenmodes of certain fibers. The introduction of astigmatic
elements in optical resonators or after them leads to the generation of
Hermite-Laguerre-Gaussian and Ince-Gaussian beams \cite{Schwarzetal04}. The
Laguerre-Gaussian beams are also proposed for the applications in free-space
optical communications systems, where the information is encoded as orbital
angular momentum states of the beam \cite{Gibsonetal04}, in quantum optics to
design entanglement states of photons \cite{Mairetal01},
\cite{MolinaTerrizaetal07}, in laser ablation \cite{Hamazakietal10}, and in
optical metrology \cite{Furhapteretal05}, to name a few examples. Angular
momentum of laser modes is discussed in \cite{YaoPadgett11}.

\subsection{Bessel-Gaussian Beams}

Use of the familiar generating relations%
\begin{align}
&  \sum_{n=0}^{\infty}\frac{L_{n}^{\alpha}(\xi)\ t^{n}}{\Gamma\left(
\alpha+n+1\right)  }=(\xi t)^{-\alpha/2}e^{t}J_{\alpha}\left(  2\sqrt{\xi
t}\right) \label{BesselGaussian2D}\\
&  \quad=\frac{e^{t}}{\Gamma(\alpha+1)}\ _{0}F_{1}\left(  -;\alpha+1;-\xi
t\right)  ,\quad J_{\nu}(z)=\frac{\left(  z/2\right)  ^{\nu}}{\Gamma(\nu
+1)}~_{0}F_{1}\left(
\begin{array}
[c]{c}%
-\medskip\\
\nu+1
\end{array}
;-\frac{z^{2}}{4}\right) \nonumber
\end{align}
in (\ref{2DGaussLaguerreSpreading}) results in a new multi-parameter family of
the Bessel-Gaussian beams:%
\begin{align}
&  B(x,y,s)=\frac{1}{1+2\alpha_{0}s+i\beta_{0}^{2}s}\exp\left(  -is\frac
{\left.  \delta_{0}^{(1)}\right.  ^{2}+\left.  \delta_{0}^{(2)}\right.  ^{2}%
}{2\left(  1+2\alpha_{0}s+i\beta_{0}^{2}s\right)  }\right)
\label{BesselGaussian}\\
&  \times\exp\left(  i\frac{\left(  2\alpha_{0}+i\beta_{0}^{2}\right)  \left(
x^{2}+y^{2}\right)  +2\left(  \delta_{0}^{(1)}+i\beta\varepsilon_{0}%
^{(1)}\right)  x+2\left(  \delta_{0}^{(2)}+i\beta\varepsilon_{0}^{(2)}\right)
y}{2\left(  1+2\alpha_{0}s+i\beta_{0}^{2}s\right)  }\right) \nonumber\\
&  \times\exp\left(  \frac{2\beta_{0}s\left(  \delta_{0}^{(1)}\varepsilon
_{0}^{(1)}+\delta_{0}^{(2)}\varepsilon_{0}^{(2)}\right)  -\left(
1+2\alpha_{0}s\right)  \left(  \left.  \varepsilon_{0}^{(1)}\right.
^{2}+\left.  \varepsilon_{0}^{(2)}\right.  ^{2}\right)  }{2\left(
1+2\alpha_{0}s+i\beta_{0}^{2}s\right)  }\right) \nonumber\\
&  \times\exp\left(  t\frac{1+2\alpha_{0}s-i\beta_{0}^{2}s}{1+2\alpha
_{0}s+i\beta_{0}^{2}s}\right)  \left(  \frac{\beta_{0}\left(  x+iy-\left(
\delta_{0}^{(1)}+i\delta_{0}^{(2)}\right)  s\right)  +\left(  \varepsilon
_{0}^{(1)}+i\varepsilon_{0}^{(2)}\right)  \left(  1+2\alpha_{0}s\right)
}{1+2\alpha_{0}s+i\beta_{0}^{2}s}\right)  ^{m}\nonumber\\
&  \times\ _{0}F_{1}\left(
\begin{array}
[c]{c}%
-\medskip\\
m+1
\end{array}
;-t\frac{\left(  \beta_{0}\left(  x-\delta_{0}^{(1)}s\right)  +\varepsilon
_{0}^{(1)}\left(  1+2\alpha_{0}s\right)  \right)  ^{2}+\left(  \beta
_{0}\left(  y-\delta_{0}^{(2)}s\right)  +\varepsilon_{0}^{(2)}\left(
1+2\alpha_{0}s\right)  \right)  ^{2}}{\left(  1+2\alpha_{0}s+i\beta_{0}%
^{2}s\right)  ^{2}}\right)  .\nonumber
\end{align}
(See~\cite{notebook} for an automatic verification.) For the complex-valued
parameters, among two interesting special cases are multi-parameter
\textquotedblleft elegant\textquotedblright\ Bessel-Gaussian beams, when
$2\alpha_{0}=i\beta_{0}^{2}:$
\begin{align}
&  B^{\left(  \text{el}\right)  }(x,y,s)=\frac{1}{1+2i\beta_{0}^{2}s}%
\exp\left(  -is\frac{\left.  \delta_{0}^{(1)}\right.  ^{2}+\left.  \delta
_{0}^{(2)}\right.  ^{2}}{2\left(  1+2i\beta_{0}^{2}s\right)  }\right)
\label{ElegantB}\\
&  \times\exp\left(  \frac{t-\beta_{0}^{2}\left(  x^{2}+y^{2}\right)  -\left(
\beta\varepsilon_{0}^{(1)}-i\delta_{0}^{(1)}\right)  x-\left(  \beta
\varepsilon_{0}^{(2)}-i\delta_{0}^{(2)}\right)  y}{1+2i\beta_{0}^{2}s}\right)
\nonumber\\
&  \times\exp\left(  \frac{2\beta_{0}s\left(  \delta_{0}^{(1)}\varepsilon
_{0}^{(1)}+\delta_{0}^{(2)}\varepsilon_{0}^{(2)}\right)  -\left(  1+i\beta
_{0}^{2}s\right)  \left(  \left.  \varepsilon_{0}^{(1)}\right.  ^{2}+\left.
\varepsilon_{0}^{(2)}\right.  ^{2}\right)  }{2\left(  1+2i\beta_{0}%
^{2}s\right)  }\right) \nonumber\\
&  \times\left(  \frac{\beta_{0}\left(  x+iy-\left(  \delta_{0}^{(1)}%
+i\delta_{0}^{(2)}\right)  s\right)  +\left(  \varepsilon_{0}^{(1)}%
+i\varepsilon_{0}^{(2)}\right)  \left(  1+i\beta_{0}^{2}s\right)  }%
{1+2i\beta_{0}^{2}s}\right)  ^{m}\nonumber\\
&  \times\ _{0}F_{1}\left(
\begin{array}
[c]{c}%
-\medskip\\
m+1
\end{array}
;-t\frac{\left(  \beta_{0}\left(  x-\delta_{0}^{(1)}s\right)  +\varepsilon
_{0}^{(1)}\left(  1+i\beta_{0}^{2}s\right)  \right)  ^{2}+\left(  \beta
_{0}\left(  y-\delta_{0}^{(2)}s\right)  +\varepsilon_{0}^{(2)}\left(
1+i\beta_{0}^{2}s\right)  \right)  ^{2}}{\left(  1+2i\beta_{0}^{2}s\right)
^{2}}\right)  ,\nonumber
\end{align}
and multi-parameter \textquotedblleft diffraction-free\textquotedblright%
\ Bessel beams, when $2\alpha_{0}=-i\beta_{0}^{2}:$%
\begin{align}
&  B^{\left(  \text{dif}\right)  }(x,y,s)=\exp\left(  t\left(  1-2i\beta
_{0}^{2}s\right)  -\left(  \beta\varepsilon_{0}^{(1)}-i\delta_{0}%
^{(1)}\right)  x-\left(  \beta\varepsilon_{0}^{(2)}-i\delta_{0}^{(2)}\right)
y\right) \label{DifB}\\
&  \times\exp\left(  \beta_{0}s\left(  \delta_{0}^{(1)}\varepsilon_{0}%
^{(1)}+\delta_{0}^{(2)}\varepsilon_{0}^{(2)}\right)  -\frac{\left(
1-i\beta_{0}^{2}s\right)  }{2}\left(  \left.  \varepsilon_{0}^{(1)}\right.
^{2}+\left.  \varepsilon_{0}^{(2)}\right.  ^{2}\right)  -\frac{is}{2}\left(
\left.  \delta_{0}^{(1)}\right.  ^{2}+\left.  \delta_{0}^{(2)}\right.
^{2}\right)  \right) \nonumber\\
&  \times\left(  \beta_{0}\left(  x+iy-\left(  \delta_{0}^{(1)}+i\delta
_{0}^{(2)}\right)  s\right)  +\left(  \varepsilon_{0}^{(1)}+i\varepsilon
_{0}^{(2)}\right)  \left(  1-i\beta_{0}^{2}s\right)  \right)  ^{m}\nonumber\\
&  \times\ _{0}F_{1}\left(
\begin{array}
[c]{c}%
-\medskip\\
m+1
\end{array}
;-t\left(  \beta_{0}\left(  x-\delta_{0}^{(1)}s\right)  +\varepsilon_{0}%
^{(1)}\left(  1-i\beta_{0}^{2}s\right)  \right)  ^{2}-t\left(  \beta
_{0}\left(  y-\delta_{0}^{(2)}s\right)  +\varepsilon_{0}^{(2)}\left(
1-i\beta_{0}^{2}s\right)  \right)  ^{2}\right)  .\nonumber
\end{align}
For $m=0$ and $\varepsilon_{0}^{(1,2)}=0,$ the latter beams have the peculiar
property of conserving the same disturbance distribution, apart from the phase
factor, across any plane parallel to the $xy$-plane in the direction of
propagation: $x=x_{0}+\delta_{0}^{(1)}s,$ $y=y_{0}+\delta_{0}^{(2)}s,$
$z=z_{0}+s.$ Graphical examples are given in \cite{notebook}.

Diffraction-free Bessel beams are reviewed in \cite{AbramAlievaRodrigo12},
\cite{TuruFrieb09} (see also \cite{Durnin87}, \cite{DurninMiceli87},
\cite{Gorietal87}, \cite{Sprangehafizi91} and the references therein for
classical accounts on propagation-invariant optical fields and Bessel modes).

\subsection{Spiral Beams}

Two-dimensional solutions of the paraxial wave equation \eqref{2DFree}, that
possess the propagation-invariant property%
\[
\iint_{\mathbb{R}^{2}}\,\big\vert B(x,y,0)\big\vert^{2}\ dx\,dy=\iint%
_{\mathbb{R}^{2}}\,\big\vert B(X,Y,s)\big\vert^{2}\ dX\,dY=\text{constant}%
\]
under rotation and rescaling $X=\rho(s)\left(  x\cos\theta(s)+y\sin
\theta\left(  s\right)  \right)  ,$ $Y=\rho(s)\left(  -x\sin\theta
(s)+y\sin\theta\left(  s\right)  \right)  ,$ were investigated in detail
\cite{AbramVolUFN}, \cite{AbramVolBook}, \cite{PiesSCHechSha00}, and
\cite{Schecheretal96}.

In Section~\ref{sec:BreathingSpiral}, we have already analyzed the transition
to a rotating frame of reference; see Equations
\eqref{Rotation}--\eqref{FockLaguerre2D}. As a combined result,
Equation~\eqref{2DFree} by means of the substitution
\begin{align}
&  B(x,y,s)=\frac{1}{\left(  s^{2}+1\right)  ^{1/2}}\exp\left(  \frac
{is\left(  x^{2}+y^{2}\right)  }{2\left(  s^{2}+1\right)  }\right)
\ \label{2DLinQuadRot}\\
&  \quad\times C\left(  \frac{x\cos\left(  \omega\arctan s\right)
+y\sin\left(  \omega\arctan s\right)  }{\sqrt{s^{2}+1}},\frac{-x\sin\left(
\omega\arctan s\right)  +y\cos\left(  \omega\arctan s\right)  }{\sqrt{s^{2}%
+1}},\ \arctan s\right) \nonumber
\end{align}
can be transformed into the equation of motion for the isotropic planar
harmonic oscillator in a perpendicular uniform magnetic field, namely:%
\begin{equation}
2iC_{s}+C_{xx}+C_{yy}=\left(  x^{2}+y^{2}\right)  C+2i\omega\left(
xC_{y}-yC_{x}\right)  \label{2DOscilMag}%
\end{equation}
(our transformation (\ref{2DLinQuad}) can be thought of as its special case
when $\omega=0).$ An algorithmic derivation is provided in~\cite{notebook}.

A straightforward use of Fock's solutions (\ref{FockLaguerre2D}) does not lead
directly to a new family of spiral beams due to the cancelation of the
crucial parameter~$\omega$ (see Section~\ref{sec:LaguerreGaussian} in the
\textsl{Mathematica} notebook~\cite{notebook}). For example, the solution
\begin{equation}
B(x,y,s)=\frac{e^{-i\left(  m+2n+1\right)  \arctan s}}{\sqrt{s^{2}+1}}%
\exp\left(  -\frac{x^{2}+y^{2}}{2\left(  1+is\right)  }\right)  \left(
\frac{x+iy}{\sqrt{s^{2}+1}}\right)  ^{m}L_{n}^{m}\left(  \frac{x^{2}+y^{2}%
}{s^{2}+1}\right)  ,\quad m\geq0 \label{SpiralSimple}%
\end{equation}
is verified by a direct substitution \cite{notebook}. (A multi-parameter
extension can be obtained by the action of Schr\"{o}dinger's group.)

On the second thought, with the help of (\ref{2DLinQuadRot}), we shall look
for a spiral beam in the form:%
\begin{equation}
B(x,y,s)=\frac{1}{\left(  s^{2}+1\right)  ^{1/2}}\exp\left(  \frac{is\left(
x^{2}+y^{2}\right)  }{2\left(  s^{2}+1\right)  }\right)  C(X,Y,T).
\label{SpiralBeamsExpansion}%
\end{equation}
Here, a familiar eigenfunction expansion \cite{AbramVolUFN},
\cite{AbramVolBook}:%
\begin{equation}
C(X,Y,T) = \sum_{n\geq0} \sum_{m\geq0} c_{n,m}^{(\pm)}\ \mathcal{C}%
_{n,m}^{(\pm)}(X,Y,T) , \label{SpiralBeamExpansionC}%
\end{equation}
in terms of Laguerre-Gaussian modes, must satisfy the auxiliary equation
\eqref{2DOscilMag}. In complex form, $z=1+is,$ $T=\arg z=\arctan s,$ and%
\begin{equation}
Z=X+iY=\frac{x+iy}{\left\vert z\right\vert }e^{-i\omega\arg z},\qquad
X=\operatorname{Re}Z,\quad Y=\operatorname{Im}Z. \label{ComplexXY}%
\end{equation}
Denoting for $m\geq0,$%
\begin{equation}
\mathcal{C}=\mathcal{C}_{n,m}^{(\pm)}(X,Y,T) =e^{-ikT}\left(  X\pm iY\right)
^{m}e^{-\left\vert Z\right\vert ^{2}/2}L_{n}^{m}\left(  \left\vert
Z\right\vert ^{2}\right)  , \label{ComplexLaguerre}%
\end{equation}
we obtain an important \textquotedblleft eigenfunction
identity\textquotedblright:%
\begin{align}
&  2i\mathcal{C}_{T}+\mathcal{C}_{XX}+\mathcal{C}_{YY}-\left(  X^{2}%
+Y^{2}\right)  \mathcal{C}-2i\omega\left(  X\mathcal{C}_{Y}-Y\mathcal{C}%
_{X}\right) \label{EigenValueC}\\
&  \qquad=2\left(  k\pm m\omega-m-2n-1\right)  \mathcal{C}\nonumber
\end{align}
by a direct evaluation \cite{notebook}.

As a result, substituting the series (\ref{SpiralBeamExpansionC}) into
Equation~\eqref{2DOscilMag}, one gets
\[
\sum_{n\geq0}\sum_{m\geq0}c_{n,m}^{(\pm)}\left(  k\pm m\omega-m-2n-1\right)
\mathcal{C}_{n,m}^{(\pm)}(X,Y,T)=0
\]
or, in view of the completeness of the Laguerre-Gaussian modes,%
\begin{equation}
c_{n,m}^{(\pm)}\left(  k\pm m\omega-m-2n-1\right)  =0.
\label{SpiralBeamConditions}%
\end{equation}
Nontrivial solutions of this equation and the corresponding spiral beams are
analyzed in the original works \cite{AbramVolUFN}, \cite{AbramVolBook}. A
multi-parameter extension can be obtained by the action of Schr\"{o}dinger's group.

\subsection{\textquotedblleft Smart\textquotedblright\ Lens Design}
\label{sec:smart}

A lens can be used to focus a laser beam to a small spot, or to produce a beam
of suitable diameter and phase structure for injection into a given optical
device \cite{AhmanNikBook}, \cite{KlineKay65}, \cite{KogelnikLi66}, \cite{Luneburg64},
\cite{VinRudSuxBook79}, \cite{WaltherLENSES}. The
multi-parameter modes under consideration allow one to adapt a required lens
design in paraxial optics to the given field configuration. For instance, in
the co-dimensional $1D$ case, let us consider the Gaussian package (\ref{1DLinNewSols}) when
$n=0.$ We found in Section~\ref{sec:HG} that the focal point is given by%
\begin{equation}
x_{0}=-\frac{2\alpha_{0}\delta_{0}+\beta_{0}^{3}\varepsilon_{0}}{4\alpha
_{0}^{2}+\beta_{0}^{4}},\qquad s_{0}=-\frac{2\alpha_{0}}{4\alpha_{0}^{2}%
+\beta_{0}^{4}}, \label{1Dfocal}%
\end{equation}
say for $\alpha_{0}\geq0.$ In this section, let us consider a lens-like medium
with quadratic refractive index \cite{KogelnikLi66}, as in Equation
(\ref{Schroudinger}), on $(0,l)$ such that our solutions (\ref{AFunction}%
)--(\ref{hhK}) can be used on this interval and the continuity condition holds
at $s=0.$ For the region $s\geq l,$ one can take the Gaussian package in
(\ref{1DLinNewSols}), once again, but with $s\rightarrow s-l,$ $\alpha
_{0}\rightarrow\alpha(l),$ etc. due to Equations (\ref{hhA})--(\ref{hhK}),
which automatically implies the field continuity at $s=l.$ Moreover, critical
points of the intensity inside the lens occur when%
\begin{align}
&  x_{\text{min},\text{max}}=-\frac{\varepsilon_{0}}{\beta_{0}}\cos\left(
s_{\text{min},\text{max}}\right)  +\left(  \delta_{0}-\frac{2\alpha
_{0}\varepsilon_{0}}{\beta_{0}}\right)  \sin\left(  s_{\text{min},\text{max}%
}\right)  ,\label{MaxMin}\\
&  \qquad\tan\left(  2s_{\text{min},\text{max}}\right)  =\frac{4\alpha_{0}%
}{1-4\alpha_{0}^{2}-\beta_{0}^{4}}.\nonumber
\end{align}
Here, $\alpha\left(  s_{\text{min},\text{max}}\right)  =0,$ thus providing a
one-to-one correspondence with the minimum-uncertainty squeezed states in
quantum mechanics \cite{Kr:Sus12}, \cite{KrySusVegaMinimum}. For the length of
lens we choose: $0\leq s_{\text{min}}<l<s_{\text{max}}.$ Then, location of the
beam focal point in the homogeneous medium, after passing through the lens, is
given by%
\begin{equation}
x_{f}=-\frac{2\alpha(l)\delta(l)+\beta^{3}(l)\varepsilon(l)}{4\alpha
^{2}(l)+\beta^{4}(l)},\qquad s_{f}-l=-\frac{2\alpha(l)}{4\alpha^{2}%
(l)+\beta^{4}(l)}. \label{1Dfocalf}%
\end{equation}
As a result, in view of the invariant \cite{KrySusVegaMinimum},%
\begin{equation}
\frac{4\alpha^{2}+\beta^{4}+1}{\beta^{2}}=\frac{4\alpha_{0}^{2}+\beta_{0}%
^{4}+1}{\beta_{0}^{2}}, \label{InvariantLens}%
\end{equation}
one arrives at the following relation between two focal points,%
\begin{align}
1-\frac{2\alpha\delta+\beta^{3}\varepsilon}{x_{f}}  &  =\left(  \frac{\beta
}{\beta_{0}}\right)  ^{2}\left(  1-\frac{2\alpha_{0}\delta_{0}+\beta_{0}%
^{3}\varepsilon_{0}}{x_{0}}\right)  ,\label{XF}\\
1-\frac{2\alpha}{s_{f}-l}  &  =\left(  \frac{\beta}{\beta_{0}}\right)
^{2}\left(  1-\frac{2\alpha_{0}}{s_{0}}\right)  . \label{SF}%
\end{align}
(Here, the corresponding solutions (\ref{hhA})--(\ref{hhK}) are evaluated at
the terminal point of the lens, $s=l.)$ Finally,%
\begin{equation}
\frac{1}{\beta_{0}^{2}}+\frac{1}{r_{0}^{2}}=\frac{1}{\beta^{2}}+\frac{1}%
{r_{f}^{2}}, \label{RF}%
\end{equation}
for the beam radii at the focal points before, $(x_{0},s_{0}),$ and after,
$(x_{f},s_{f}),$ the lens, respectively. According to our analysis, location
of the terminal focal point $(x_{f},s_{f})$ and the corresponding beam radius
$r_{f}$ both depend on the lens length,\ $l,$ which can be adjusted for an
\textquotedblleft optimal control\textquotedblright\ of the beam propagation
through this lens in a certain optical device. For instance, by
(\ref{InvariantLens}) and (\ref{RF}),%
\begin{equation}
\frac{2\beta_{0}^{2}}{r_{f}^{2}}=4\alpha_{0}^{2}+\beta_{0}^{4}+1+\left(
4\alpha_{0}^{2}+\beta_{0}^{4}-1\right)  \cos2l-4\alpha_{0}\sin2l,
\label{minmax}%
\end{equation}
which attains its minimum and maximum values, namely,
\begin{equation}
\text{max},\text{min}\left(  \frac{2\beta_{0}^{2}}{r^{2}}\right)  =4\alpha
_{0}^{2}+\beta_{0}^{4}+1\pm\sqrt{\left(  4\alpha_{0}^{2}+\left(  \beta_{0}%
^{2}+1\right)  ^{2}\right)  \left(  4\alpha_{0}^{2}+\left(  \beta_{0}%
^{2}-1\right)  ^{2}\right)  }, \label{MINmax}%
\end{equation}
at the critical points $l=s_{\text{max},\text{min}},$ respectively (cf.
\cite{KrySusVegaMinimum}). An important ratio,%
\begin{equation}
\frac{r_{\text{min}}^{2}}{r_{\text{max}}^{2}}=\frac{4\alpha_{0}^{2}+\beta
_{0}^{4}+1-\sqrt{\left(  4\alpha_{0}^{2}+\left(  \beta_{0}^{2}+1\right)
^{2}\right)  \left(  4\alpha_{0}^{2}+\left(  \beta_{0}^{2}-1\right)
^{2}\right)  }}{4\alpha_{0}^{2}+\beta_{0}^{4}+1+\sqrt{\left(  4\alpha_{0}%
^{2}+\left(  \beta_{0}^{2}+1\right)  ^{2}\right)  \left(  4\alpha_{0}%
^{2}+\left(  \beta_{0}^{2}-1\right)  ^{2}\right)  }}, \label{SuperFocus}%
\end{equation}
defines the maximum possible compression (or superfocusing in the terminology
of Refs.~\cite{Demkov09}, \cite{DemkovMeyer04}) of the beam inside of a
(sufficiently long) lens.

\textit{A numerical example}. Visualization of the intensity and optical
energy distribution generated by a sample lens, when $\alpha_{0}=2\beta
_{0}=2,$ $s_{\text{min}}=1.33897,$ $l=2.75,$ $s_{\text{max}}=2.90977,$
$x_{0}=x_{f}=0,$ $s_{0}=-0.23529,$ $s_{f}=2.9106,$ and $r_{0}=0.242536,$
$r_{f}=0.239104,$ is given in Figure~\ref{fig.lens}; see also our supplementary material
\cite{notebook} for more details.

\begin{figure}[htbp]
\centering \includegraphics[width=0.84\textwidth]{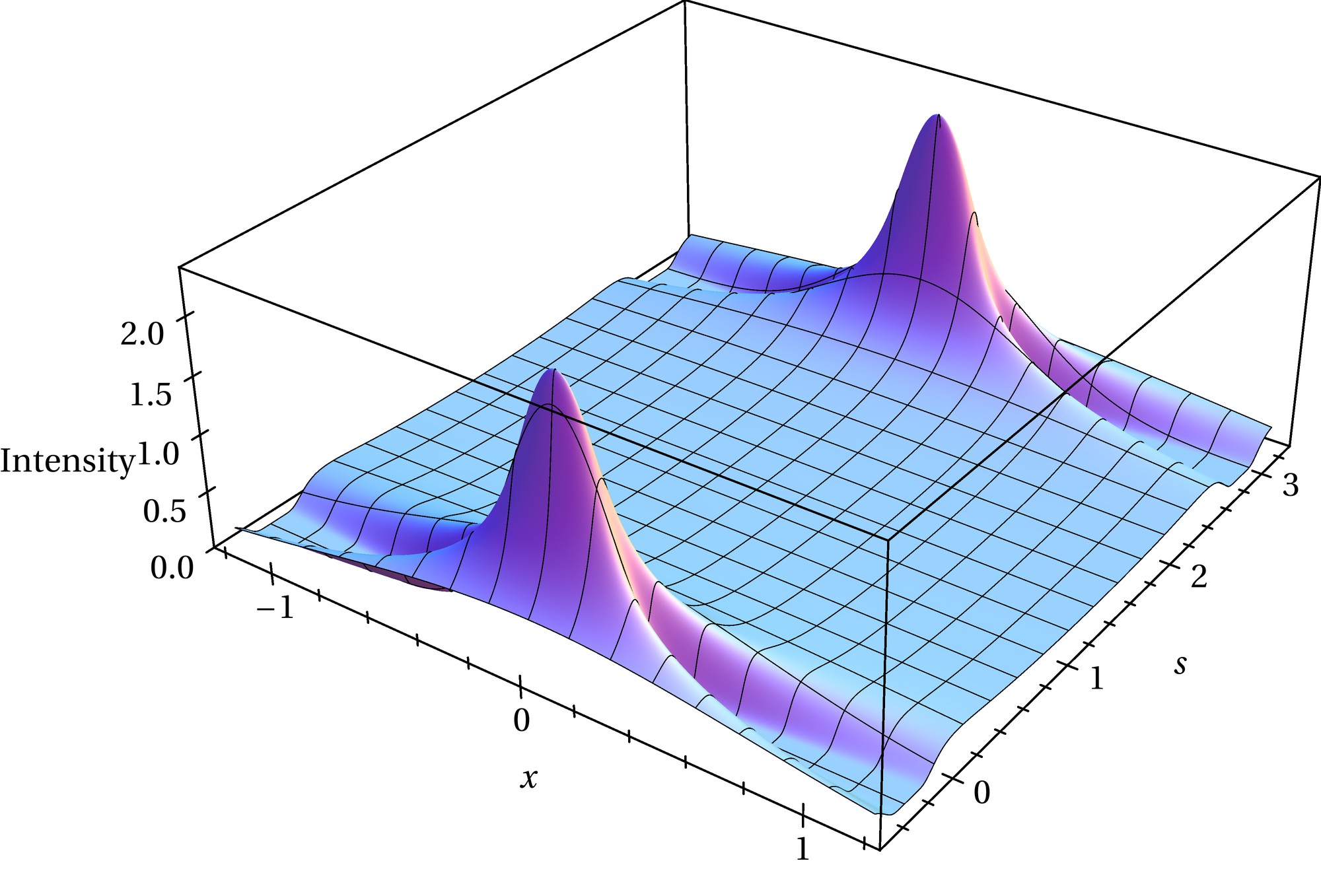}
\caption{Propagation of the beam through the sample lens under consideration.}
\label{fig.lens}
\end{figure}

\subsection{Applications to Quantum Mechanics}

A similar effect of the superfocusing of proton beam in a thin monocrystal
film was discussed in \cite{Demkov09}, \cite{DemkovMeyer04} as certain
dynamical manipulations with the system of particles by transforming high
concentration in the momentum space (collinearity) into the concentration in
the coordinate space (focusing) (validity of the $2D$ harmonic crystal model
had been confirmed by Monte Carlo computer experiments \cite{DemkovMeyer04}).
In the context of quantum mechanics, the corresponding solutions represent the
minimum-uncertainty, or squeezed, states of harmonic oscillators that are
explicitly given in \cite{KrySusVegaMinimum}. Our sample lens from the
previous section can be thought of as a codimension $1D$ model of Demkov's microscope.

Among other quantum mechanical analogs, the minimum-uncertainty squeezed
states for atoms and photons in a cavity, are reviewed in
\cite{KrySusVegaMinimum}; some of them were experimentally realized on cold
trapped atoms \cite{LeibfriedetalWineland03}. It is worth noting that similar
states can be identified for the motion of a quantum particle in a uniform
magnetic field \cite{Fock28-1}.

\subsection{Extensions to Nonlinear Paraxial Optics}

For high-intensity beams, nonlinear medium effects should be taken into
account in the theory of wave propagation. See \cite{Desyatnikovetal10},
\cite{GurevichNLIonosphere}, \cite{KuzTur85}, \cite{Mah:Sus12},
\cite{MahalovSuslov12a}, \cite{Mah:Sus13}, \cite{PonArg06}, \cite{Tao09},
\cite{VinRudSuxBook79}, \cite{VlasovTalanovParabolicEq} and the references
therein for extensions to nonlinear geometrical optics. A generalization of
Lemma~\ref{lem:nonlinear} for combination of certain nonlinear terms is
discussed in \cite{Mah:Sus13} but search for solutions of nonlinear equations
is much more complicated.

In the $1D$ linear case, where nonspreading Airy beams were introduced
\cite{BerryBalazs79} (see also \cite{SiviloglouChris07},
\cite{Siviloglouetal07}), the symmetry of the free Schr\"{o}dinger equation
can be used in order to obtain multi-parameter solutions (\ref{MultiAiry}).
Although the corresponding $1D$ cubic nonlinear Schr\"{o}dinger equation is no
longer preserved under the expansion transformation (but has a similarity
reduction to the second Painlev\'{e} equation \cite{GagWint89},
\cite{GiaJos89}, \cite{Mah:Sus12}, \cite{Smith76}, \cite{Tajiri83}), the same
symmetry holds for the quintic nonlinear Schr\"{o}dinger equation, which is
thus invariant under the action of this group. Here, the blow up, namely a
singularity such that the wave amplitude tends to infinity in a finite time,
occurs (see \cite{MahalovSuslov12a}, \cite{Suslov11}, \cite{Tao09} and the
references therein).

As is well known, a similar symmetry holds for the homogeneous $2D$ cubic
nonlinear Schr\"{o}dinger equation \cite{KuzTur85}, \cite{Talanov70} (in
optics this symmetry is known as Talanov's transformation
\cite{VlasovTalanovParabolicEq}). This is another classical example of the
blow up phenomenon. The stationary $2D$ waveguides in homogeneous quadratic
Kerr media are unstable \cite{KuzTur85}. Under certain conditions,
self-focusing of light beams occurs on a finite distance despite diffraction
spreading. Moreover, for parabolic channels in a monocrystal film, the cubic
nonlinearity may further enhance superfocusing of particle beams predicted in
\cite{Demkov09}, \cite{DemkovMeyer04}. The corresponding inhomogeneous medium
effects deserve a detailed study. An extension to randomly varying media is
also of interest (cf. \cite{AhDyaChirBook81}, \cite{RytovetalPrinciples},
\cite{TangMah13}).

\section{Computer Algebra Methods}
\label{sec:CA}

For an automatic verification of the results presented in this paper, we used
the computer algebra system \textsl{Mathematica}, and in some specific
instances, the \texttt{HolonomicFunctions} package~\cite{Koutschan10b},
written by the first-named author in the frame of his Ph.D.
thesis~\cite{Koutschan09}\footnote{The package can be downloaded from
http:/$\!$/www.risc.jku.at/research/combinat/software/HolonomicFunctions/}.
(See also \cite{KoutschanPauleSuslov14} and the references therein for
applications of the \texttt{HolonomicFunctions} package to relativistic
Coulomb integrals.)

The application of computer algebra in the context of the present paper comes
in three different flavors: The first one employs Gr\"obner bases, the second
one is based on the built-in simplification procedures of \textsl{Mathematica}%
, and the third one is related to the above-mentioned
\texttt{HolonomicFunctions} package.

Gr\"obner bases were introduced in~\cite{Buchberger65} and are a very useful
tool for computations with polynomial ideals. For finding ``nice'' expressions
for the solutions~\eqref{A0}--\eqref{F0} of the Riccati-type system, one can
consider the ideal generated by the (polynomial) equations
\eqref{SysA}--\eqref{SysF}. Equivalence of expressions then corresponds to
equality modulo the ideal. See~\cite{notebook} for more details.

Similarly, we discovered an \textquotedblleft invariant\textquotedblright\ of
the Ermakov-type system. Again using Equations \eqref{SysA}--\eqref{SysF} as
input (but now with $c_{0}=1$) one can use Gr\"{o}bner bases to find relations
that are implied by the given equations. Searching for an equation that does
not involve the parameters $a,b,c,d,f,g$ yields the identity
\begin{equation}
\beta^{2}\kappa^{\prime}-\beta\delta\varepsilon^{\prime}+(\delta^{2}+\beta
^{2}\varepsilon^{2})\gamma^{\prime}=0
\end{equation}
which was missing in the original publications. It reveals that the
differential equations in the Ermakov-type system are in fact dependent. In
particular, Equation~\eqref{SysF} for $\kappa^{\prime}$ can be derived from
the previous equations of this system.

To demonstrate the other two applications, recall the multi-parameter Airy
modes $B(x,s)$ given in Equation~\eqref{MultiAiry}. Thanks to the progress
that computer algebra systems like \textsl{Mathematica} have been made during
the past decades, particularly in dealing with special functions, it can be
directly verified that $B(x,s)$ satisfies the parabolic
equation~\eqref{FreeSchroedinger}: one just inputs the expression given on the
right-hand side of \eqref{MultiAiry} and differentiates it symbolically. Then
the command \texttt{FullSimplify} successfully simplifies the expression
$iB_{s}+B_{xx}$ to~$0$, see the corresponding section in the accompanying
notebook~\cite{notebook}.

The last approach achieves more, and is a bit of an overhead if one only
wanted to verify that $B(x,s)$ satisfies the given differential equation.
Namely, the \texttt{HolonomicFunctions} package computes the set of all
differential equations that a given expression satisfies (more precisely: a
finite basis of this, in general, infinite set). For the multi-parameter Airy
modes, the software computes the following two differential equations:
\begin{align}
(4 \alpha s+1)^{2} B_{s} + 2 p_{1} B_{x} - i p_{2} B  &  = 0,\\
(4 \alpha s+1)^{2} B_{xx} -2 i p_{1} B_{x} - p_{2} B  &  = 0,
\end{align}
where the polynomial coefficients $p_{1}$ and $p_{2}$ are given by
\begin{align}
p_{1}  &  = \delta+4 \alpha\delta s+\beta^{3} s+8 \alpha^{2} s x+2 \alpha x,\\
p_{2}  &  = 2 i \alpha+ \beta^{2} \varepsilon+ \delta^{2}+8 i\alpha^{2} s+4
\alpha^{2} x^{2}+4 \alpha\delta x+ \beta^{3} x,
\end{align}
and where $\alpha=\alpha(0)$ etc. Obviously, the parabolic equation
$iB_{s}+B_{xx}=0$ is just a simple linear combination of the above two
equations. Thus, we again have proved that $B(x,s)$ satisfies $iB_{s}%
+B_{xx}=0$, but even more: the program has found this equation automatically,
starting from the closed form of its solution as the sole input.

Similarly, the remaining formulas in this paper can be verified and/or
derived. For the holonomic systems approach to work, some inputs have to be
transformed into an appropriate format, e.g., the expression given by
\eqref{AFunction}--\eqref{hhK}: holonomic functions are closed under addition,
multiplication, and substitution of algebraic expressions. Since $\sin(s)$ and
$\cos(s)$, which appear in the argument of the Hermite polynomials, are not
algebraic, one may apply the transformation $s\mapsto i\log(z)$ in order to
turn the trigonometric functions into rational functions. More details and all
other computations are contained in the accompanying \textsl{Mathematica}
notebook~\cite{notebook}.

\section{Conclusion}

This work is dedicated to a mathematical description of light propagation in
turbid media and/or through optical systems that are subject to a natural
noisy environment. To this end, we apply concepts of the Fresnel diffraction,
the generalized lens transformation, see Lemma~\ref{lem:nonlinear}, and
computer algebra tools \cite{Koutschan09}, \cite{Koutschan10b},
\cite{Koutschan13} in order to analyze multi-parameter families of certain
propagation-invariant laser beams in codimension $1D$ and $2D$ configurations that are
important in paraxial optics and its applications. Independent proofs of these
results are provided in the supplementary electronic material~\cite{notebook}
along with a computer algebra verification of all related mathematical tools
introduced in the original publications without sufficient details. In
summary, the \textquotedblleft missing\textquotedblright\ multi-parameter
solutions of the paraxial wave equations, that are studied in this article,
allow one to describe all main features of the special laser modes propagation
in a variety of optical systems, in a consistent mathematical way, with the
help of a computer algebra system. In numerical simulations, we have concentrated
on results which are potentially of a practical and academic value, e.g.,
for software development and pedagogy.

\noindent{\textbf{Acknowledgements}}. This research was partially carried out
during our participation in the Summer School on \textquotedblleft
Combinatorics, Geometry and Physics\textquotedblright\ at the Erwin
Schr\"{o}dinger International Institute for Mathematical Physics (ESI),
University of Vienna, in June 2014. We wish to express our gratitude to
Christian Krattenthaler for his hospitality. The first-named author was
supported by the Austrian Science Fund (FWF): W1214, the second-named author
by the Simons
Foundation Grant~\#316295 and by the National Science Foundation Grant
DMS-1440664, and the third-named author by the AFOSR Grant
FA9550-11-1-0220. We are grateful to Eugeny G. Abramochkin, Sergey I. Kryuchkov,
Vladimir I. Man'ko, and Peter Paule for valuable comments and to Miguel A. Bandres for kindly
pointing out the reference \cite{BandresGuiz09} to our attention.
Suggestions from the referees are much appreciated. Last but not least,
we would like to thank Aleksei P. Kiselev for communicating the interesting
articles \cite{Kiselev07}, \cite{KiselevPerel}, \cite{Kiselevetal12}.

\appendix

\section{From Maxwell to Paraxial Wave}

We follow \cite{AbramVolBook} with somewhat different details. In dielectrics
(no free current, no free charge, isotropic, homogeneous, material linear),
the Maxwell equations for the complex electric $\boldsymbol{E}$ and magnetic
$\boldsymbol{H}$ fields for a monochromatic wave varying as $e^{-i\omega t}$
are given by%
\begin{equation}
\operatorname{curl}\boldsymbol{E}=i\frac{\omega}{c}\mu\boldsymbol{H}%
,\qquad\ \ \operatorname{div}\left(  \mu\boldsymbol{H}\right)  =0,
\label{Maxwell12}%
\end{equation}%
\begin{equation}
\operatorname{curl}\boldsymbol{H}=-i\frac{\omega}{c}\varepsilon\boldsymbol{E}%
,\qquad\operatorname{div}\left(  \varepsilon\boldsymbol{E}\right)  =0,
\label{Maxwell34}%
\end{equation}
where $\varepsilon$ is the permittivity and $\mu$ is the permeability of the
material (see, for example, \cite{AhmanNikBook}, \cite{Laxetal75},
\cite{VainshBook}, \cite{VinRudSuxBook79}). Let us consider a
\textquotedblleft polarized\textquotedblright\ wave of the form,%
\begin{equation}
\boldsymbol{E}=f(x,y,z)e^{ikz}\boldsymbol{e}_{x}+g(x,y,z)e^{ikz}%
\boldsymbol{e}_{z},\qquad k^{2}=\varepsilon\mu\frac{\omega^{2}}{c^{2}},
\label{Electric}%
\end{equation}
where $\left\{  \boldsymbol{e}_{x},\boldsymbol{e}_{y},\boldsymbol{e}%
_{z}\right\}  $ are orthonormal vectors in $\left.
%TCIMACRO{\U{211d} }%
%BeginExpansion
\mathbb{R}
%EndExpansion
\right.  ^{3}.$ From the first Equation (\ref{Maxwell12}) one gets:%
\begin{align}
&  i\frac{\omega}{c}\mu\boldsymbol{H}=\operatorname{curl}\boldsymbol{E}%
\label{Magnetic}\\
&  \quad=\frac{\partial g}{\partial y}e^{ikz}\boldsymbol{e}_{x}+\left(
\frac{\partial f}{\partial z}+ikf-\frac{\partial g}{\partial x}\right)
e^{ikz}\boldsymbol{e}_{y}-\frac{\partial f}{\partial y}e^{ikz}\boldsymbol{e}%
_{z}\nonumber
\end{align}
and the second Equation (\ref{Maxwell12}) is automatically satisfied. In
addition, from the second Equation (\ref{Maxwell34}):%
\begin{equation}
\frac{\partial f}{\partial x}+\frac{\partial g}{\partial z}+ikg=0.
\label{Divergence}%
\end{equation}
If $g\equiv0,$ then $f_{x}=0$ and the only transversal\ solution is a plane
wave, $\boldsymbol{E}=e^{ikz}\boldsymbol{e}_{x},$ up to a constant multiple
(cf. \cite{Laxetal75}).

In a similar fashion, from the first Equation (\ref{Maxwell34}) and
(\ref{Magnetic}) we obtain:%
\begin{align}
&  k^{2}\boldsymbol{E}=\operatorname{curl}\left(  i\frac{\omega}{c}%
\mu\boldsymbol{H}\right)  =\left(  k^{2}f+g_{xz}+ikg_{x}-f_{yy}-2ikf_{z}%
-f_{zz}\right)  e^{ikz}\boldsymbol{e}_{x}\label{Prelim}\\
&  +\left(  f_{xy}+g_{yz}+ikg_{y}\right)  e^{ikz}\boldsymbol{e}_{y}+\left(
f_{xz}+ikf_{x}-g_{xx}-g_{yy}\right)  e^{ikz}\boldsymbol{e}_{z}.\nonumber
\end{align}
In view of (\ref{Divergence}), the latter equation can be simplified to%
\begin{align*}
&  k^{2}\boldsymbol{E}=k^{2}\left(  f\boldsymbol{e}_{x}+g\boldsymbol{e}%
_{z}\right)  e^{ikz}\\
&  \quad=\left(  k^{2}f-f_{xx}-f_{yy}-2ikf_{z}-f_{zz}\right)  e^{ikz}%
\boldsymbol{e}_{x}+\left(  k^{2}g-g_{xx}-g_{yy}-2ikg_{z}-g_{zz}\right)
e^{ikz}\boldsymbol{e}_{z}.
\end{align*}
Finally, under the imposed conditions $\left\vert f_{zz}\right\vert
\ll2\left\vert kf_{z}\right\vert $ and $\left\vert g_{zz}\right\vert
\ll2\left\vert kg_{z}\right\vert ,$ we arrive to the paraxial wave equations,%
\begin{equation}
2ik\frac{\partial F}{\partial z}+\frac{\partial^{2}F}{\partial x^{2}}%
+\frac{\partial^{2}F}{\partial y^{2}}=0, \label{ParaxialWaveEqF}%
\end{equation}
for the transversal and longitudinal components, $F=\left\{  f,g\right\}  ,$
of the complex electric field, respectively. (The corresponding magnetic field
can be evaluated by (\ref{Magnetic}).) Once again, these components are
related by (\ref{Divergence}), which implies that%
\begin{equation}
\frac{\partial}{\partial z}\left(  e^{ikz}g\right)  =-f_{x}e^{ikz},
\end{equation}
and, integrating by parts,%
\[
e^{ikz}g=-\int f_{x}e^{ikz}\ dz=-\frac{1}{ik}f_{x}e^{ikz}+\frac{1}{ik}\int
f_{xz}e^{ikz}\ dz\thickapprox-\frac{1}{ik}f_{x}e^{ikz},
\]
provided that $\left\vert k\right\vert \gg1.$ In paraxial approximation, it is
a custom to write%
\begin{equation}
g\thickapprox-\frac{1}{ik}f_{x}=-\frac{1}{ik}\frac{\partial f}{\partial x}
\label{Longitude}%
\end{equation}
for the small longitudinal component of electric field that automatically
satisfies (\ref{ParaxialWaveEqF}). More details can be found in
\cite{AbramVolBook}, \cite{Laxetal75}. A general solution is a superposition
of two waves of the form (\ref{Electric}).

\textit{Note}. The paraxial wave equations (\ref{ParaxialWaveEqF}) for
transversal and longitudinal components, $F=\left\{
f(x,y,z),g(x,y,z)\right\}  ,$ can be solved by the Fresnel integral,%
\begin{equation}
F(x,y,z)=\frac{k}{2\pi i z}\iint_{\mathbb{R}^{2}}\exp\left(  \frac{ik}%
{2z}\left[  (x-\xi)^{2}+(y-\eta)^{2}\right]  \right)  F_{0}(\xi,\eta)\ d\xi
d\eta, \label{FresnelF}%
\end{equation}
subject to proper \textquotedblleft initial\textquotedblright\ data,
$F_{0}=\left\{  f_{0}(x,y),g_{0}(x,y)\right\}  ,$ which are related as
follows,%
\begin{equation}
g_{0}+\frac{1}{ik}\frac{\partial f_{0}}{\partial x}+\frac{1}{2k^{2}}\left(
\frac{\partial^{2}g_{0}}{\partial x^{2}}+\frac{\partial^{2}g_{0}}{\partial
y^{2}}\right)  =0, \label{fgcondition0}%
\end{equation}
in view of \textquotedblleft divergence\textquotedblright\ condition
(\ref{Divergence}). (When $k\gg1,$ one formally gets (\ref{Longitude}).)

In fact, Equation (\ref{fgcondition0}) is the $2D$ inhomogeneous Helmholtz
equation \cite{SommPDE}, \cite{VladPDE}:%
\begin{equation}
\frac{\partial^{2}g_{0}}{\partial x^{2}}+\frac{\partial^{2}g_{0}}{\partial
y^{2}}+2k^{2}g_{0}=2ik\frac{\partial f_{0}}{\partial x}, \label{Helmholtz}%
\end{equation}
which can be solved exactly provided that function $f_{0}(x,y)$ is known.
Under the Sommerfeld radiation condition,%
\begin{equation}
\lim_{r\rightarrow\infty}r^{1/2}\left(  \frac{\partial}{\partial r}-ik\sqrt
{2}\right)  g_{0}\left(  r\boldsymbol{e}\right)  =0,\quad r=\sqrt{x^{2}+y^{2}%
}, \label{SommCond}%
\end{equation}
uniformly in $\boldsymbol{e},$ $\left\vert \boldsymbol{e}\right\vert =1,$ one
gets \cite{Couto13}, \cite{VladPDE}:%
\begin{equation}
g_{0}(x,y)=\frac{k}{2}\iint_{\mathbb{R}^{2}}H_{0}^{(1)}\left(  k\sqrt{2\left[
(x-\zeta)^{2}+(y-\vartheta)^{2}\right]  }\right)  \frac{\partial f_{0}%
}{\partial\zeta}(\zeta,\vartheta)\ d\zeta d\vartheta, \label{HelmholtzHankel}%
\end{equation}
where $H_{0}^{(1)}(z)$ is a Hankel function \cite{NiUvSF}.

\section{From Maxwell to Nonlinear Paraxial Optics}

In a more general case (of a weakly inhomogeneous linear or nonlinear medium
with a complex-valued dielectric permittivity $\varepsilon;$ see, for example,
Refs.~\cite{Fock65} and \cite{VainshBook} for more details),
one can look for solutions of Equations (\ref{Maxwell12})--(\ref{Maxwell34})
as a superposition,%
\begin{equation}
\boldsymbol{E}=\boldsymbol{E}_{x}+\boldsymbol{E}_{y},\label{EGeneral}%
\end{equation}
of two \textquotedblleft polarized\textquotedblright\ waves:%
\begin{align}
\boldsymbol{E}_{x} &  =f(x,y,z)e^{ik(z)}\boldsymbol{e}_{x}+g(x,y,z)e^{ik(z)}%
\boldsymbol{e}_{z},\label{EXY}\\
\boldsymbol{E}_{y} &  =h(x,y,z)e^{ik(z)}\boldsymbol{e}_{y}+l(x,y,z)e^{ik(z)}%
\boldsymbol{e}_{z},\nonumber
\end{align}
where $f,$ $g,$ $h,$ $l,$ and $k$ are some complex-valued functions. In a
similar fashion,%
\begin{align}
&  F_{xx}+F_{yy}+F_{zz}+2ik_{z}F_{z}+\left(  \varepsilon\mu\frac{\omega^{2}%
}{c^{2}}-k_{z}^{2}+ik_{zz}\right)  F\label{FE}\\
&  \quad\quad=-\left\{
\begin{array}
[c]{c}%
\left(  \mathcal{E}+\dfrac{\varepsilon_{z}}{\varepsilon}\left(  g+l\right)
\right)  _{x}\medskip\\
\left(  \mathcal{E}+\dfrac{\varepsilon_{z}}{\varepsilon}\left(  g+l\right)
\right)  _{y}%
\end{array}
\right.  ,\quad\mathcal{E}=\frac{\varepsilon_{x}}{\varepsilon}f+\frac
{\varepsilon_{y}}{\varepsilon}h\nonumber
\end{align}
and%
\begin{align}
&  G_{xx}+G_{yy}+G_{zz}+2ik_{z}G_{z}+\left(  \varepsilon\mu\frac{\omega^{2}%
}{c^{2}}-k_{z}^{2}+ik_{zz}\right)  G\label{GE}\\
&  \qquad\qquad\qquad\qquad=-ik_{z}\left(  \mathcal{E+}\frac{\varepsilon_{z}%
}{\varepsilon}G\right)  -\left(  \mathcal{E+}\frac{\varepsilon_{z}%
}{\varepsilon}G\right)  _{z},\nonumber
\end{align}
where, by definition,%
\begin{equation}
F=\left\{
\begin{array}
[c]{c}%
f\medskip(x,y,z)\\
h(x,y,z)
\end{array}
\right.  ,\qquad G=\left\{
\begin{array}
[c]{c}%
g\medskip(x,y,z)\\
l(x,y,z)
\end{array}
\right.  .\label{FG}%
\end{equation}
Here, it is convenient to rewrite the last equation (\ref{Maxwell34}) as a sum
of two equations:%
\begin{equation}
f_{x}+g_{z}+ik_{z}g+\frac{\varepsilon_{x}}{\varepsilon}f+\frac{\varepsilon
_{z}}{\varepsilon}g=0,\quad h_{y}+l_{z}+ik_{z}l+\frac{\varepsilon_{y}%
}{\varepsilon}h+\frac{\varepsilon_{z}}{\varepsilon}l=0.\label{fghl}%
\end{equation}
We did not impose any conditions yet and Equations (\ref{EGeneral}%
)--(\ref{fghl}) are equivalent to the original Maxwell system
(\ref{Maxwell12})--(\ref{Maxwell34}) under consideration. For paraxial
approximation, we may choose $k_{zz}=0,$ namely, $k(z)=kz,$ where $k$ is a constant.

Let us first consider linear and nonlinear codimension $1D$ cases. When $h=l=f_{y}%
=g_{y}=\varepsilon_{y}=0,$ one can simplify to%
\begin{align}
&  f_{xx}+f_{zz}+2ikf_{z}+\left(  \varepsilon\mu\frac{\omega^{2}}{c^{2}}%
-k^{2}\right)  f+\left(  \frac{\varepsilon_{x}}{\varepsilon}f+\dfrac
{\varepsilon_{z}}{\varepsilon}g\right)  _{x}=0,\label{2fg}\\
&  g_{xx}+g_{zz}+2ikg_{z}+\left(  \varepsilon\mu\frac{\omega^{2}}{c^{2}}%
-k^{2}\right)  g=-ik\left(  \frac{\varepsilon_{x}}{\varepsilon}f\mathcal{+}%
\frac{\varepsilon_{z}}{\varepsilon}g\right)  -\left(  \frac{\varepsilon_{x}%
}{\varepsilon}f\mathcal{+}\frac{\varepsilon_{z}}{\varepsilon}g\right)
_{z},\label{2gf}\\
&  f_{x}+g_{z}+ikg+\frac{\varepsilon_{x}}{\varepsilon}f+\frac{\varepsilon_{z}%
}{\varepsilon}g=0. \label{1fg}%
\end{align}
From the last equation,%
\begin{equation}
f=-\frac{e^{-ikz}}{\varepsilon}\int(\varepsilon ge^{ikz})_{z}\ dx,\qquad
g=-\frac{e^{-ikz}}{\varepsilon}\int e^{ikz}(\varepsilon f)_{x}\ dz.
\label{fgInts}%
\end{equation}
Thus Equations (\ref{2fg}) and (\ref{2gf}) can be thought of as certain
integro-differential equations for complex-valued functions $f$ and $g,$
respectively. Integrating by parts,%
\begin{equation}
g=-\frac{e^{-ikz}}{ik\varepsilon}\int(\varepsilon f)_{x}\ de^{ikz}%
=-\frac{(\varepsilon f)_{x}}{ik\varepsilon}+\frac{e^{-ikz}}{ik\varepsilon}\int
e^{ikz}(\varepsilon f)_{xz}\ dz\approx-\frac{(\varepsilon f)_{x}%
}{ik\varepsilon}. \label{gbyprts}%
\end{equation}
For large $|k|,$ it is also a custom to assume that $\left\vert f_{zz}%
\right\vert \ll2\left\vert kf_{z}\right\vert ,$ $\left\vert f_{zz}\right\vert
\ll2\left\vert kf_{z}\right\vert ,$ and $|g|\ll|f|.$ As a result, one may
concentrate on the study of scalar inhomogeneous paraxial wave equation of the
form:%
\begin{equation}
f_{xx}+2ikf_{z}+\left(  \varepsilon\mu\frac{\omega^{2}}{c^{2}}+\left(
\frac{\varepsilon_{x}}{\varepsilon}\right)  _{x}-k^{2}\right)  f+\frac
{\varepsilon_{x}}{\varepsilon}f_{x}=0. \label{ParaxialWavefEq}%
\end{equation}
In a weakly inhomogeneous nonlinear medium, we expand the (complex-valued)
permittivity $\varepsilon,$
\begin{equation}
\varepsilon(x,z)=\varepsilon_{0}(z)+\varepsilon_{1}(z)x+\varepsilon
_{2}(z)x^{2}+...\left(  +\lambda\left\vert f\right\vert ^{2}+...\right)
\label{1Depsilon}%
\end{equation}
and neglect the higher order terms. In this approximation,%
\begin{equation}
\frac{\varepsilon_{x}}{\varepsilon}=\frac{\varepsilon_{1}}{\varepsilon_{0}%
}+\left[  2\frac{\varepsilon_{2}}{\varepsilon_{0}}-\left(  \frac
{\varepsilon_{1}}{\varepsilon_{0}}\right)  ^{2}\right]  x,\qquad\left(
\frac{\varepsilon_{x}}{\varepsilon}\right)  _{x}=2\frac{\varepsilon_{2}%
}{\varepsilon_{0}}-\left(  \frac{\varepsilon_{1}}{\varepsilon_{0}}\right)
^{2}. \label{DersEpsilon}%
\end{equation}
and one arrives at a form of the paraxial wave equation
(\ref{SchroedingerQuadratic}) (or its nonlinear versions).

The corresponding linear and nonlinear codimension $2D$ cases, when one can concentrate on
a certain dominant component of electric field once again, are similar.
Further details are left to the reader.

\end{document}